\documentclass[letterpaper,journal]{IEEEtran}

\usepackage{amsmath,amssymb,amsfonts,bm}
\usepackage{amsthm}
\usepackage{cite}
\usepackage{array}
\usepackage[caption=false,font=normalsize,labelfont=sf,textfont=sf]{subfig}
\usepackage{textcomp}
\usepackage{stfloats}
\usepackage{url}
\usepackage{graphicx}
\usepackage{xcolor}
\usepackage{placeins}
\usepackage{float}
\usepackage[colorlinks=true,linkcolor=blue,citecolor=blue,urlcolor=blue,bookmarksopen,bookmarksnumbered]{hyperref}
\newtheoremstyle{ieee-result}
  {3pt}{3pt}{\itshape}{}{\bfseries}{.}{0.5em}
  {\thmname{#1}\thmnumber{ #2}\thmnote{\normalfont\space(#3)}}
\theoremstyle{ieee-result}
\newtheorem{assumption}{Assumption}
\newtheorem{definition}{Definition}
\newtheorem{lemma}{Lemma}
\newtheorem{proposition}{Proposition}
\newtheorem{theorem}{Theorem}
\newtheorem{corollary}{Corollary}
\newtheorem{remark}{Remark}

\newcommand{\CN}{\mathbb{C}}

\newcommand{\EISL}{\operatorname{EISL}}
\newcommand{\ESL}{\operatorname{ESL}}
\newcommand{\ISL}{\operatorname{ISL}}

\newcommand{\tr}{\operatorname{Tr}}
\begin{document}

\title{Statistical Analysis of 3D Ambiguity Functions for Random MIMO-OFDM Signals}

\author{Pinjun Zeng,~\IEEEmembership{Student Member,~IEEE},
        Yuyang Lu,~\IEEEmembership{Student Member,~IEEE},
        Jialan Tang,~\IEEEmembership{Student Member,~IEEE},
        Xiaoyang Chen,~\IEEEmembership{Student Member,~IEEE}, 
        Yifeng Xiong,~\IEEEmembership{Member,~IEEE},
        Fan Liu,~\IEEEmembership{Senior Member,~IEEE},
        and Shi Jin,~\IEEEmembership{Fellow,~IEEE}
        \thanks{Fan Liu is the corresponding author.}}

\maketitle

\begin{abstract}
Random communication symbols allow integrated sensing and communications (ISAC) systems to exploit payload signals for sensing, but the resulting matched-filter response varies with the transmitted data. Existing statistical analyses of sensing with random orthogonal frequency division multiplexing (OFDM) waveforms have largely focused on delay and Doppler, while the spatial dimension introduced by multi-input multi-output (MIMO) arrays and precoding has received much less attention. This paper develops a three-dimensional ambiguity function (3D-AF) for random MIMO-OFDM waveforms over delay, Doppler, and angle domains. For point-target sensing, we express the 3D-AF as a quadratic form of the random communication symbols, with a deterministic matrix that captures the delay and Doppler shifts, the transmit and receive array responses, and the precoders on different subcarriers. Based on this representation, we derive a closed-form expression for the expected squared 3D-AF and show explicitly how the constellation fourth-order moment (or kurtosis) affects the ambiguity sidelobes. The same characterization yields expected integrated sidelobe level (EISL) and expected sidelobe level (ESL) metrics over mixed delay-Doppler-angle search domains. The framework is further extended to coherent multi-symbol processing with slow-time Doppler, leading to the corresponding fast-slow time (FST) approximation. These results provide a statistical characterization of the sensing ambiguity induced by random MIMO-OFDM communication signals in delay, Doppler, and angle.
\end{abstract}


\section{Introduction}

\IEEEPARstart{I}{ntegrated} sensing and communications (ISAC) is regarded as an enabling technology for future wireless networks, in which communication and sensing functions share spectrum, hardware, and signaling resources~\cite{sturm2011fusion}.
By allowing communication waveforms to probe the surrounding environment, ISAC can provide sensing capability without dedicated radar signals or separate sensing infrastructure.
This resource-sharing principle is particularly relevant to emerging sixth-generation (6G) applications that require both high-rate connectivity and environmental awareness~\cite{liu2022isac}.

OFDM and MIMO transmission are central components of modern wireless systems and are therefore natural candidates for ISAC\@.
OFDM offers a flexible time-frequency resource grid, whereas MIMO arrays provide spatial degrees of freedom for beamforming, multiplexing, and angle estimation~\cite{zhang2021overview}.
For millimeter-wave (mmWave) joint radar-communications systems, wide bandwidths and large antenna arrays also motivate joint waveform design and spatial processing~\cite{mishra2019mmwave}.
Multicarrier ISAC studies have investigated time-frequency resource design~\cite{crossdomain2024} and waveform optimization under resource-occupancy constraints~\cite{resourceoccupancy2024}.
Modern communication signaling schemes contain a variety of deterministic reference signals, such as demodulation reference signals (DMRS) and channel state information reference signals (CSI-RS), which can be directly exploited for sensing.
However, these signals occupy only a limited fraction of the radio resources.
To further improve sensing performance, payload symbols occupying a substantial part of the resource grid can also be exploited, although they vary with the transmitted data~\cite{liu2022isac}.

Using random payload symbols for sensing makes the matched-filter response dependent on the transmitted data realization.
Consequently, the ambiguity sidelobes also fluctuate across different payload realizations and are naturally characterized from a statistical perspective.
The ambiguity function (AF) describes the matched-filter response under delay and Doppler mismatch and has been widely used to evaluate radar sensing performance~\cite{woodward1953,levanon2004}.
Statistical studies of random communication waveforms have characterized delay-Doppler AFs~\cite{zhang2025discreteaf,inputdist2024} and the zero-Doppler autocorrelation sidelobes of pulse-shaped signals~\cite{liao2024pulseshaping}.
Sparse vector coding provides a complementary waveform-design approach to suppress delay-Doppler sidelobes while embedding communication data~\cite{svc2022}.
For MIMO-OFDM signals, however, the matched-filter response also varies with angle and with the precoder used on each subcarrier.
This motivates a statistical AF characterization over the joint delay-Doppler-angle domain.

MIMO-OFDM sensing has been studied from several perspectives.
Existing ISAC works have considered waveform and precoder design under communication and sensing requirements~\cite{temiz2021precoders,constructive2024,slp2025}.
Related designs maximize radar detection probability under peak-to-average power ratio (PAPR) and communication constraints~\cite{papr2025,detectionprob2024}.
MIMO radar research has considered transmit beampattern and target-return cross-correlation design~\cite{stoica2007mimo}, as well as AF formulations that account for array geometry and transmitted waveforms, including transmit beamspace processing~\cite{sanantonio2007mimoaf,li2015beamspace}.
These studies have established that both time-frequency signaling and spatial processing may strongly affect sensing performance.
For MIMO-OFDM communication signals carrying random payload data, the ambiguity surface varies with the transmitted symbols, even when the precoders and array configuration remain fixed.
A statistical characterization is therefore needed to describe this variability.

Recent work on sensing with random communication waveforms has begun to address this problem explicitly. The expected ranging sidelobes have been expressed as a closed-form function of the modulation basis, pulse shaping, and the fourth-order moment of the transmitted symbols \cite{liu2025iceberg,liu2025cpofdm}. This dependence has motivated constellation shaping for OFDM sensing~\cite{du2024reshaping}. Separately, precoding for random MIMO signals has been studied using an ergodic sensing estimation-error criterion~\cite{lu2024randomprecoding}. More relevant to this paper, Zhang et al.~\cite{zhang2025discreteaf} developed a statistical analysis of the discrete periodic AF for random communication waveforms and derived closed-form expected sidelobe level (ESL) and expected integrated sidelobe level (EISL) expressions on the delay-Doppler plane, together with a fast-slow time (FST) approximation. The analysis in \cite{zhang2025discreteaf}, however, does not account for the angular dimension or MIMO precoding. Extending this statistical framework to MIMO-OFDM therefore requires a joint treatment of delay, Doppler, angle, and the spatial processing applied across subcarriers.

In this paper, we formulate the delay-Doppler-angle AF for random MIMO-OFDM communication waveforms with spatial processing at both the transmitter and receiver. We show that the 3D-AF admits a quadratic form in the random payload symbols, where the corresponding matrix is determined by the delay-Doppler shift, array responses, and MIMO precoding. Based on this quadratic form, we derive a closed-form expression for the expected squared 3D-AF and the associated sidelobe metrics. Unlike delay and Doppler, whose effects depend on the offset from the true parameter under the adopted cyclic model, the angular response also depends on the target direction itself through the array manifold. This distinction is taken into account throughout the 3D-AF formulation and its statistical analysis. In light of the above framework, the main contributions of this paper are summarized as follows.

\begin{itemize}
  \item We formulate a 3D delay-Doppler-angle AF for random MIMO-OFDM communication waveforms and relate it directly to the matched-filter output for point-target sensing. We further show that the 3D-AF admits a quadratic form in the random payload symbols, where the corresponding matrix accounts for the delay-Doppler shift, array responses, and MIMO precoding.
  \item We derive a closed-form expression for the expected squared 3D-AF under independent identically distributed (i.i.d.) communication symbols. The expression shows explicitly how the fourth-order moment of the symbol constellation affects the ambiguity sidelobes, and further yields the EISL and ESL over a three-dimensional delay-Doppler-angle search region.
  \item We extend the analysis to coherent processing over multiple OFDM symbols using the FST model. Under the explicit small intra-block Doppler-phase condition, we retain the slow-time Doppler bin $q\in\mathbb{Z}_M$ and show that the resulting AF matrix has a Kronecker structure when the same precoders are adopted across OFDM symbols.
\end{itemize}

The remainder of this paper is organized as follows. Section II presents the ULA, random-symbol, MIMO-OFDM transmit, and point-target echo models, followed by the single-symbol and FST multi-symbol 3D-AF definitions and sidelobe metrics. Section III develops the quadratic-form statistical characterization for both cases. Section IV compares the theoretical expected squared FST-AF on the discrete delay-Doppler cells and the continuous angular interval, with angular samples used for numerical integration, and Section V concludes the paper. Appendices provide the detailed derivations.

\emph{Notations}: Matrices are denoted by bold uppercase letters, e.g., $\mathbf{A}$; vectors are denoted by bold lowercase letters, e.g., $\mathbf{s}$; and scalars are denoted by normal font, e.g., $N$. The $i$th entry of $\mathbf{s}$ is denoted by $s_i$, and the $(m,n)$th entry of $\mathbf{A}$ is denoted by $[\mathbf{A}]_{m,n}$. The transpose, Hermitian transpose, and entry-wise complex conjugate are denoted by $(\cdot)^T$, $(\cdot)^H$, and $(\cdot)^*$, respectively. The trace, expectation, and covariance operators are denoted by $\operatorname{Tr}(\cdot)$, $\mathbb{E}[\cdot]$, and $\operatorname{Cov}(\cdot,\cdot)$, respectively. The Kronecker product is denoted by $\otimes$. The notation $\operatorname{diag}(\mathbf{a})$ denotes a diagonal matrix with diagonal entries given by $\mathbf{a}$, and $\operatorname{blkdiag}(\mathbf{A}_1,\ldots,\mathbf{A}_M)$ denotes a block-diagonal matrix. The identity matrix of size $N$ is denoted by $\mathbf{I}_N$. The Frobenius inner product is defined as $\langle \mathbf{A},\mathbf{B}\rangle_F := \operatorname{Tr}(\mathbf{A}^H\mathbf{B})$, and the Frobenius norm is denoted by $\|\cdot\|_F$.

\section{System Model and Performance Metrics}\label{sec:system_model}

This section establishes the signal and sensing models and performance metrics used in the
subsequent statistical analysis. We first construct the random MIMO-OFDM waveform and its array
response, following the above derivation. Then we formulate point-target matched filtering and define the single-symbol 3D-AF and its
sidelobe metrics, and finally introduce the multi-symbol and FST models.

Consider a MIMO-OFDM ISAC system with $N$ subcarriers, $N_t$ transmit antennas, $N_r$ receive
antennas, and $N_s$ data streams on each subcarrier. The wireless propagation channel is assumed to be narrowband and located in the far-field region. Both the transmit and receive antenna arrays are configured as uniform linear arrays (ULAs) with half-wavelength inter-element spacing. The length of the cyclic prefix (CP) is set to be no less than the maximum excess delay defined in the channel model.
Based on these system settings, the single-symbol three-dimensional ambiguity function (3D-AF) is defined over a mixed discrete-continuous domain: integer cyclic-delay and Doppler frequency grids, and a continuous angular interval.
\subsection{Random MIMO-OFDM Signal and Array Model}

\begin{definition}[ULA Steering Vectors]\label{def:steering_vectors}
  For a direction $\theta$ measured from array broadside, the transmit and receive steering vectors are defined as
  \begin{align*}
    \mathbf{a}_t(\theta)
     & :=
    \frac{1}{\sqrt{N_t}}
    \left[
      1, e^{j\pi\sin\theta}, \ldots,
      e^{j\pi(N_t-1)\sin\theta}
      \right]^T
    \in\CN^{N_t}, \\
    \mathbf{a}_r(\theta)
     & :=
    \frac{1}{\sqrt{N_r}}
    \left[
      1, e^{j\pi\sin\theta}, \ldots,
      e^{j\pi(N_r-1)\sin\theta}
      \right]^T
    \in\CN^{N_r}.
  \end{align*}
\end{definition}

With the array responses specified, we next introduce the random data symbols carried by the OFDM
resource grid and construct the corresponding array-domain transmit waveform.

Consider an OFDM resource grid with $N$ subcarriers. On the $n$th subcarrier, let
$\mathbf{s}_n\in\CN^{N_s}$ denote the corresponding random data-symbol vector. Stacking these
vectors over all subcarriers gives
\begin{equation}
  \mathbf{s}
  :=
  \begin{bmatrix}
    \mathbf{s}_0^T, \mathbf{s}_1^T, \cdots, \mathbf{s}_{N-1}^T
  \end{bmatrix}^T
  \in\CN^{NN_s}.
\end{equation}
\begin{assumption}[Random Symbol Statistics]\label{ass:symbol_statistics}
  The scalar data symbols on all subcarriers and streams are proper, mutually independent, and
  identically distributed with unit average power. In particular, for every scalar symbol $s_i$,
  \begin{equation}
    \mathbb{E}\left[s_i\right]=0,\quad
    \mathbb{E}\left[|s_i|^2\right]=1,\quad
    \mathbb{E}\left[s_i^2\right]=0.
    \label{eq:symbol_assumptions}
  \end{equation}
\end{assumption}

\begin{definition}[Constellation Kurtosis]\label{def:constellation_kurtosis}
  Under Assumption~\ref{ass:symbol_statistics}, the kurtosis of the constellation is defined as
  \begin{equation}
    \kappa:=\frac{\mathbb{E}\!\left[|s_i-\mathbb{E}\{s_i\}|^4\right]}
    {\mathbb{E}\!\left[|s_i-\mathbb{E}\{s_i\}|^2\right]^2}
    =\mathbb{E}\!\left[|s_i|^4\right].
    \label{eq:constellation_kurtosis}
  \end{equation}
\end{definition}

In particular, QPSK and higher-order uniform PSK are proper and have constant modulus, and thus
$\kappa=1$, whereas QAM has non-constant modulus but still satisfies $1<\kappa<2$. The kurtosis
values of typical proper PSK and QAM constellations are listed in Table~\ref{tab:kurtosis}~\cite{liu2025cpofdm}.

\begin{table}[H]
  \caption{Kurtosis Values of Typical Proper PSK and QAM Constellations}
  \label{tab:kurtosis}
  \centering
  \footnotesize
  \renewcommand{\arraystretch}{1.12}
  \setlength{\tabcolsep}{4.2pt}
  \begin{tabular}{lcccc}
    \hline\hline
    Constellation & Proper PSK & 16-QAM  & 64-QAM   & 128-QAM  \\
    \hline
    Kurtosis      & 1          & 1.32    & 1.381    & 1.3427   \\
    \hline
    Constellation & 256-QAM    & 512-QAM & 1024-QAM & 2048-QAM \\
    \hline
    Kurtosis      & 1.3953     & 1.3506  & 1.3988   & 1.3525   \\
    \hline\hline
  \end{tabular}
\end{table}

Let $\mathbf{W}_n\in\CN^{N_t\times N_s}$ denote the digital precoder on the $n$th subcarrier. The
corresponding spatially precoded transmit vector is
\begin{equation}
  \mathbf{v}_n
  :=
  \mathbf{W}_n\mathbf{s}_n
  \in\CN^{N_t}.
\end{equation}
\indent Collecting these vectors column by column produces the frequency-domain transmit matrix
\begin{equation}
  \mathbf{V}
  :=
  \begin{bmatrix}
    \mathbf{v}_0, \mathbf{v}_1, \cdots, \mathbf{v}_{N-1}
  \end{bmatrix}
  \in\CN^{N_t\times N},
\end{equation}
whose rows correspond to transmit antennas and whose columns correspond to subcarriers. Let
$\mathbf{F}_N\in\CN^{N\times N}$ denote the normalized discrete Fourier transform (DFT) matrix,
\begin{equation}
  [\mathbf{F}_N]_{m,n}
  =
  \frac{1}{\sqrt{N}}e^{-j2\pi mn/N},
  \qquad
  m,n=0,\ldots,N-1.
\end{equation}
Accordingly, $\mathbf{F}_N^H$ is the IDFT matrix. Applying the IDFT along the subcarrier columns for
every transmit antenna gives the time-domain OFDM block
\begin{equation}
  \mathbf{X}
  :=
  \mathbf{V}\mathbf{F}_N^H
  \in\CN^{N_t\times N},
\end{equation}
where rows correspond to transmit antennas and columns correspond to fast-time samples.

For a given transmit angle $\theta$, the physically radiated equivalent fast-time waveform is
obtained by projecting the antenna rows of $\mathbf{X}$ onto $\mathbf{a}_t(\theta)$. Hence,
\begin{equation}
  \mathbf{x}(\theta)
  :=
  \mathbf{X}^T\mathbf{a}_t(\theta)
  = \mathbf{F}_N^H\mathbf{V}^T\mathbf{a}_t(\theta)
  \in\CN^N.
\end{equation}
The vector before the IDFT can be expanded subcarrier by subcarrier as
\begin{align*}
  \mathbf{V}^T\mathbf{a}_t(\theta)
   & =
  \begin{bmatrix}
    \mathbf{v}_0^T\mathbf{a}_t(\theta) \\
    \mathbf{v}_1^T\mathbf{a}_t(\theta) \\
    \vdots                             \\
    \mathbf{v}_{N-1}^T\mathbf{a}_t(\theta)
  \end{bmatrix}
  \nonumber \\
   & =
  \begin{bmatrix}
    \mathbf{a}_t^T(\theta)\mathbf{W}_0\mathbf{s}_0 \\
    \mathbf{a}_t^T(\theta)\mathbf{W}_1\mathbf{s}_1 \\
    \vdots                                         \\
    \mathbf{a}_t^T(\theta)\mathbf{W}_{N-1}\mathbf{s}_{N-1}
  \end{bmatrix}.
\end{align*}
Because the communication symbols have already been stacked into $\mathbf{s}$, collect the $N$
direction-dependent row mappings directly into
\begin{equation}
  \begin{aligned}
    \mathbf{G}_T(\theta)
     & :=
    \begin{bmatrix}
      \mathbf{a}_t^T(\theta)\mathbf{W}_0 & \mathbf{0} & \cdots & \mathbf{0} \\
      \mathbf{0} & \mathbf{a}_t^T(\theta)\mathbf{W}_1 & \cdots & \mathbf{0} \\
      \vdots & \vdots & \ddots & \vdots \\
      \mathbf{0} & \mathbf{0} & \cdots &
      \mathbf{a}_t^T(\theta)\mathbf{W}_{N-1}
    \end{bmatrix} \\[-1pt]
     & =
    \operatorname{blkdiag}\!\left(
      \mathbf{a}_t^T(\theta)\mathbf{W}_0,\ldots,
      \mathbf{a}_t^T(\theta)\mathbf{W}_{N-1}
    \right),
  \end{aligned}
  \label{eq:GT_def}
\end{equation}
where each diagonal block has size $1\times N_s$.
By construction, the preceding stacked relation becomes
$\mathbf{V}^T\mathbf{a}_t(\theta)=\mathbf{G}_T(\theta)\mathbf{s}$. Inserting this relation into the
directionally projected transmit waveform yields
\begin{equation}
  \mathbf{x}(\theta) = \mathbf{F}_N^H\mathbf{G}_T(\theta)\mathbf{s}.
  \label{eq:x_theta}
\end{equation}
The model permits frequency-dependent digital precoding through $\mathbf{W}_n$ and therefore
includes both frequency-flat and subcarrier-selective beamformers. Within this precoding framework, the matrix
$\mathbf{G}_T(\theta)$ is not an additional physical precoder; it is the direction-dependent
projection of the subcarrier precoders, introduced to separate the common random symbol vector from
deterministic spatial processing.

\subsection{Point-Target Sensing Model}

We consider a monostatic sensing configuration with $L$ far-field point targets. The
$\ell$th target is described by its complex scattering coefficient $\beta_\ell$ and the
three coordinates $(k_\ell,q_\ell,\theta_\ell)$, representing its cyclic delay, intra-symbol
Doppler, and direction, respectively. The receiver evaluates a candidate target template
at $(k,q,\theta)$, so the matched-filter response is determined jointly by the delay,
Doppler, and angular mismatches between the template and each target. The resulting search
domain is therefore discrete in delay and Doppler and continuous in angle.

\begin{figure}[H]
  \centering
  \includegraphics[width=0.99\linewidth]{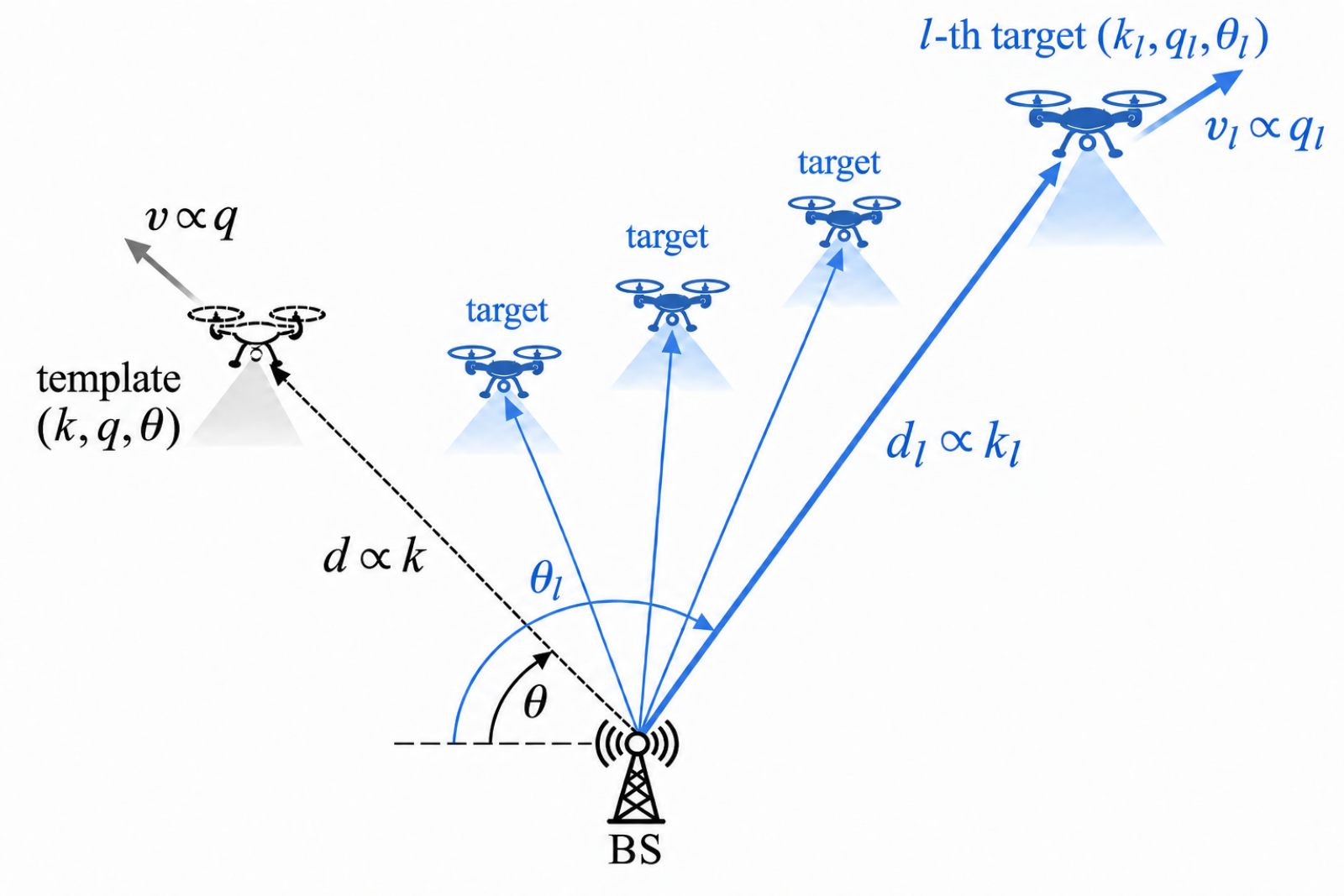}
  \caption{Geometric interpretation of the three-dimensional sensing search: a template at
    $(k,q,\theta)$ is compared with a point target at $(k_\ell,q_\ell,\theta_\ell)$, where
    the delay, Doppler, and angle coordinates determine the relative range, velocity, and
    direction, respectively.}
  \label{fig:point_target_geometry}
\end{figure}

As shown in Fig.~\ref{fig:point_target_geometry}, the template direction $\theta$ is
compared with the target direction $\theta_\ell$, while the delay and Doppler coordinates
map to the corresponding relative range and velocity. For a target at a generic direction
$\theta$, the monostatic far-field assumption gives the same departure and arrival direction,
so its deterministic narrowband spatial response is
\begin{equation}
  \mathbf{H}_s(\theta)
  :=
  \mathbf{a}_r(\theta)\mathbf{a}_t^T(\theta)
  \in\CN^{N_r\times N_t},
  \label{eq:Hs_def}
\end{equation}
where $\mathbf{a}_t(\theta)\in\CN^{N_t}$ and $\mathbf{a}_r(\theta)\in\CN^{N_r}$ are the transmit and
receive steering vectors, respectively. The rank-one matrix $\mathbf{H}_s(\theta)$ captures only the
spatial departure-arrival response; the complex scattering coefficient is kept separate in the
echo model below.

Following the discrete cyclic delay-Doppler convention in \cite{zhang2025discreteaf}, let
$k,q\in\mathbb{Z}_N:=\{0,1,\ldots,N-1\}$ denote the cyclic-delay and intra-symbol Doppler
indices, respectively. The delay operator is specified in the frequency domain by
\begin{equation}
  \boldsymbol{\Lambda}_{N,k}
  :=
  \operatorname{diag}
  \left(
  1,e^{-j2\pi k/N},\ldots,e^{-j2\pi (N-1)k/N}
  \right).
\end{equation}
The equivalent time-domain cyclic-delay matrix is
\begin{equation}
  \mathbf{J}_{N,k}
  :=
  \mathbf{F}_N^H\boldsymbol{\Lambda}_{N,k}\mathbf{F}_N
  \in\CN^{N\times N}.
  \label{eq:J_def}
\end{equation}
The intra-symbol Doppler shift is represented by
\begin{equation}
  \mathbf{D}_{N,q}
  :=
  \operatorname{diag}
  \left(
  1,e^{j2\pi q/N},\ldots,e^{j2\pi q(N-1)/N}
  \right)
  \in\CN^{N\times N}.
  \label{eq:D_def}
\end{equation}
Thus, $\boldsymbol{\Lambda}_{N,k}$ and $\mathbf{J}_{N,k}$ are frequency-domain and time-domain
representations of the same cyclic delay, whereas $\mathbf{D}_{N,q}$ is the frequency-shift matrix.

 Following the cyclic-prefix insertion and removal procedure in~\cite{zhang2025discreteaf},
 the linear delay over each OFDM block is represented by the cyclic-delay operator
 $\mathbf{J}_{N,k}$ on the CP-free samples. After CP insertion at the transmitter and CP removal
 at the receiver, superposing the returns from $L$ point targets and adding receiver noise gives
\begin{equation}
  \mathbf{Y}
  =
  \sum_{\ell=1}^{L}
  \beta_\ell
  \mathbf{H}_s(\theta_\ell)
  \mathbf{X}
  \mathbf{J}_{N,k_\ell}^T
  \mathbf{D}_{N,q_\ell}^T
  +
  \mathbf{Z},
  \label{eq:point_echo}
\end{equation}
where $\mathbf{Y},\mathbf{Z}\in\CN^{N_r\times N}$, and $\beta_\ell\in\CN$, $k_\ell$, $q_\ell$,
and $\theta_\ell$ denote the complex scattering coefficient, discrete delay, discrete fast-time
Doppler, and direction of the $\ell$th target, respectively. Both $\mathbf{X}$ and $\mathbf{Y}$ place
antenna elements along rows and fast-time samples along columns, so the delay and Doppler operators
act from the right along the fast-time dimension.

\begin{remark}[Model Scope]\label{rem:model_scope}
  The requirement that modeled delays do not exceed the CP is common in block-based OFDM
  sensing~\cite{wei2024mimoofdm,crossdomain2024,resourceoccupancy2024}.
  Here, the single-symbol 3D-AF additionally adopts the integer cyclic delay-Doppler grid
  following~\cite{zhang2025discreteaf}. Fractional delay, beam squint, and timing or carrier-synchronization errors
  unrelated to the target are outside the present model.
\end{remark}

\subsection{3D-AF and Performance Metrics}

We next define the matched-filter response and the ambiguity metrics used throughout the
statistical analysis. For a template parameter tuple $(k,q,\theta)$, the receiver uses the following
unit-scattering point-target template.

\begin{definition}[Point-Target Template and Matched-Filter Output]\label{def:point_template_mf}
  The unit-scattering point-target template and the corresponding matched-filter output are given by
  \begin{align}
    \mathbf{T}_{\rm p}(k,q,\theta)
     & :=
    \mathbf{H}_s(\theta)
    \mathbf{X}
    \mathbf{J}_{N,k}^T
    \mathbf{D}_{N,q}^T
    \in\CN^{N_r\times N},
    \label{eq:point_template} \\
    z_{\rm MF}(k,q,\theta)
     & =
    \left\langle
    \mathbf{T}_{\rm p}(k,q,\theta),\mathbf{Y}
    \right\rangle_F \nonumber\\
     & =
    \operatorname{Tr}
    \left[
      \mathbf{T}_{\rm p}^H(k,q,\theta)\mathbf{Y}
      \right].
    \label{eq:mf_output}
  \end{align}
\end{definition}
The scattering coefficient is excluded from the template because it is an unknown target amplitude
estimated from the matched-filter output. Substituting \eqref{eq:point_echo} into
\eqref{eq:mf_output} and expanding the Frobenius inner product yield
\begin{align}
  z_{\rm MF}(k,q,\theta)
   & =
  \sum_{\ell=1}^{L}
  \beta_\ell\operatorname{Tr}\!\bigl\{
    \left(
      \mathbf{H}_s(\theta)\mathbf{X}
      \mathbf{J}_{N,k}^T\mathbf{D}_{N,q}^T
    \right)^H
  \nonumber\\[-1pt]
   & \quad\times
  \mathbf{H}_s(\theta_\ell)\mathbf{X}
  \mathbf{J}_{N,k_\ell}^T\mathbf{D}_{N,q_\ell}^T
  \bigr\}
  +z_{k,q,\theta}
  \label{eq:mf_trace_expansion}\\
   & =
  \sum_{\ell=1}^{L}
  \beta_\ell
  \bigl[
    \mathbf{a}_r^H(\theta)\mathbf{a}_r(\theta_\ell)
  \bigr]
  \mathbf{a}_t^H(\theta)\mathbf{X}^*
  \mathbf{J}_{N,k}^T
  \nonumber\\[-1pt]
   & \quad\times
  \mathbf{D}_{N,q}^*\mathbf{D}_{N,q_\ell}
  \mathbf{J}_{N,k_\ell}
  \mathbf{X}^T\mathbf{a}_t(\theta_\ell)
  +z_{k,q,\theta}
  \nonumber\\
   & =
  \sum_{\ell=1}^{L}
  \beta_\ell e^{-j2\pi k_\ell(q-q_\ell)/N}
  \nonumber\\[-1pt]
   & \quad\times
  \bigl[
    \mathbf{a}_r^H(\theta)\mathbf{a}_r(\theta_\ell)
  \bigr]
  \mathbf{a}_t^H(\theta)\mathbf{X}^*
  \mathbf{J}_{N,k-k_\ell}^T
  \nonumber\\[-1pt]
   & \quad\times
  \mathbf{D}_{N,q-q_\ell}^*
  \mathbf{X}^T\mathbf{a}_t(\theta_\ell)
  +z_{k,q,\theta},
  \label{eq:mf_expansion}
\end{align}
where
$z_{k,q,\theta}:=
\langle\mathbf{T}_{\rm p}(k,q,\theta),\mathbf{Z}\rangle_F$
is the filtered noise. All delay and Doppler index differences in \eqref{eq:mf_expansion} are
interpreted modulo $N$. The final equality applies the circular delay-Doppler commutation relation
in \cite{zhang2025discreteaf} and exposes a common scalar kernel whose arguments are the
delay-Doppler mismatches and the test and target angles. This kernel motivates the following
definition.

\begin{definition}[Cross and Self 3D Ambiguity Functions]\label{def:cross_self_af}
  For a template at $(k,q,\theta)$ and a unit-scattering reference target at
  $(0,0,\vartheta)$, the point-target cross 3D delay-Doppler-angle AF is
  \begin{equation}
    \resizebox{0.90\columnwidth}{!}{$\displaystyle \mathcal{A}_{3D}(k,q;\theta,\vartheta):=\rho_r(\theta,\vartheta)\mathbf{a}_t^H(\theta)\mathbf{X}^*\mathbf{J}_{N,k}^T\mathbf{D}_{N,q}^*\mathbf{X}^T\mathbf{a}_t(\vartheta)$}.
    \label{eq:cross_af_def}
  \end{equation}
  The self 3D-AF is obtained by setting the target angle equal to the test angle, i.e.,
  \begin{equation}
    \begin{aligned}
      \mathcal{A}_{3D}(k,q,\theta)&:=\mathcal{A}_{3D}(k,q;\theta,\theta)\\[-1pt]
       &=\rho_r(\theta,\theta)\mathbf{a}_t^H(\theta)\mathbf{X}^*\mathbf{J}_{N,k}^T\mathbf{D}_{N,q}^*\mathbf{X}^T\mathbf{a}_t(\theta),
    \end{aligned}
  \end{equation}
  where $\theta$ and $\vartheta$ denote the test and target angles, respectively, and
  $\rho_r(\theta,\vartheta):=
  \mathbf{a}_r^H(\theta)\mathbf{a}_r(\vartheta)$ denotes their receive-array angular correlation.
\end{definition}

Using Definition~\ref{def:cross_self_af}, the matched-filter expansion in
\eqref{eq:mf_expansion} becomes
\begin{equation}
  \begin{aligned}
    z_{\mathrm{MF}}(k,q,\theta)
    &= \sum_{\ell=1}^{L}\beta_\ell e^{-j2\pi k_\ell(q-q_\ell)/N} \\
    &\quad \times \mathcal{A}_{3\mathrm{D}}(k-k_\ell,q-q_\ell;\theta,\theta_\ell) + z_{k,q,\theta}.
  \end{aligned}
  \label{eq:shifted_af}
\end{equation}

Thus, $k$ and $q$ act as mismatch coordinates after shifting by the true delay and Doppler cells.

\begin{remark}[Delay-Doppler Translation and Angular Dependence]\label{rem:angle_nontranslation}
  Equation~\eqref{eq:shifted_af} also clarifies the different roles of the three parameters. For a
  target in cell $(k_\ell,q_\ell,\theta_\ell)$, the delay and Doppler arguments appear as the
  differences $k-k_\ell$ and $q-q_\ell$, apart from a unit-modulus phase factor. Hence, under the
  adopted cyclic model, translating the true delay or Doppler translates the matched-filter
  response without changing its magnitude profile. The angular dependence is different: it is
  carried by $\rho_r(\theta,\theta_\ell)$ and by the pair $\mathbf{G}_T^H(\theta)$ and
  $\mathbf{G}_T(\theta_\ell)$. Because a ULA response depends on $\sin\theta$ rather than on a
  cyclic shift of $\theta$, two targets separated by the same physical angular offset need not
  produce identical cuts when their target angles differ. Thus, the angle lacks translation
  invariance.
\end{remark}

Having related the matched-filter output to the cross 3D-AF, we next isolate the random-symbol
dependence from the deterministic delay-Doppler and spatial operators. Substituting
\eqref{eq:x_theta} into \eqref{eq:cross_af_def} gives
\begin{align}
  \mathcal{A}_{3D}(k,q;\theta,\vartheta)
   & =
  \rho_r(\theta,\vartheta)
  \left(
    \mathbf{F}_N^H\mathbf{G}_T(\theta)\mathbf{s}
  \right)^H
  \mathbf{J}_{N,k}^T\mathbf{D}_{N,q}^*
  \nonumber\\[-1pt]
   & \quad\times
  \left(
    \mathbf{F}_N^H\mathbf{G}_T(\vartheta)\mathbf{s}
  \right)
  \nonumber\\
   & =
  \mathbf{s}^H\bigl[
    \rho_r(\theta,\vartheta)
    \mathbf{G}_T^H(\theta)
    \mathbf{F}_N\mathbf{J}_{N,k}^T\mathbf{D}_{N,q}^*
  \nonumber\\[-1pt]
   & \quad\times
    \mathbf{F}_N^H\mathbf{G}_T(\vartheta)
  \bigr]\mathbf{s}.
  \label{eq:quadratic_derivation}
\end{align}
The bracketed deterministic operator naturally leads to the following two definitions.

\begin{definition}[Delay-Doppler Kernel]\label{def:delay_doppler_kernel}
  The frequency-domain delay-Doppler kernel is defined by
  \begin{equation}
    \boldsymbol{\Psi}_{k,q}
    :=
    \mathbf{F}_N
    \mathbf{J}_{N,k}^T
    \mathbf{D}_{N,q}^*
    \mathbf{F}_N^H
    \in\CN^{N\times N}.
    \label{eq:Psi_def}
  \end{equation}
\end{definition}

\begin{definition}[Deterministic 3D-AF Kernel Matrix]\label{def:af_matrix}
  For a delay-Doppler cell $(k,q)$, test angle $\theta$, and target angle $\vartheta$, the
  deterministic 3D-AF kernel matrix is defined as
  \begin{equation}
    \mathbf{M}_{k,q}(\theta,\vartheta)
     :=
    \rho_r(\theta,\vartheta)\mathbf{G}_T^H(\theta)
    \boldsymbol{\Psi}_{k,q}\mathbf{G}_T(\vartheta).
    \label{eq:M_def}
  \end{equation}
\end{definition}

Combining Definitions~\ref{def:delay_doppler_kernel} and~\ref{def:af_matrix} with~\eqref{eq:quadratic_derivation} yields the quadratic-form representation
\begin{equation}
  \mathcal{A}_{3D}(k,q;\theta,\vartheta)
  =
  \mathbf{s}^H
  \mathbf{M}_{k,q}(\theta,\vartheta)
  \mathbf{s}.
\label{eq:quadratic_af}
\end{equation}

For a target at angle $\theta_0$, let
\begin{equation}
  \Omega:=\mathcal{K}\times\mathcal{Q}\times\Theta,
\end{equation}
where $\mathcal{K},\mathcal{Q}\subseteq\mathbb{Z}_N$ are finite delay and Doppler sets containing
zero and $\Theta$ is the angular search interval, expressed in radians. If
$\Theta_{\rm ml}(\theta_0)\subseteq\Theta$ is the prescribed angular mainlobe about $\theta_0$,
the mainlobe and sidelobe regions are
\begin{equation}
  \begin{aligned}
    \Omega_{\rm ml}(\theta_0)
     & :=
    \{(0,0,\theta):\theta\in\Theta_{\rm ml}(\theta_0)\}, \\
    \Omega_{\rm sl}(\theta_0)
     & :=
    \Omega\setminus\Omega_{\rm ml}(\theta_0).
  \end{aligned}
  \label{eq:sidelobe_region}
\end{equation}
For this mixed discrete-continuous domain, use the product measure $\zeta$ defined by
\begin{equation}
  \int_{\Omega} f\,\mathrm{d}\zeta
  :=
  \sum_{k\in\mathcal{K}}\sum_{q\in\mathcal{Q}}
  \int_{\Theta} f(k,q,\theta)\,\mathrm{d}\theta,
\end{equation}
and introduce the sidelobe indicator
\begin{equation}
  \begin{aligned}
    \delta_{k,q}
    & :=
    \begin{cases}
      1, & (k,q)=(0,0),\\
      0, & (k,q)\neq(0,0),
    \end{cases}\\[-1pt]
    \mathbf{1}_{\Theta_{\rm ml}(\theta_0)}(\theta)
    & :=
    \begin{cases}
      1, & \theta\in\Theta_{\rm ml}(\theta_0),\\
      0, & \theta\notin\Theta_{\rm ml}(\theta_0),
    \end{cases}\\[-1pt]
    \chi_{\rm sl}(k,q,\theta)
    & :=
    1-\delta_{k,q}\mathbf{1}_{\Theta_{\rm ml}(\theta_0)}(\theta).
  \end{aligned}
\end{equation}
where $\delta_{k,q}$ is the two-dimensional Kronecker delta that selects the matched
delay-Doppler cell. With these regions and the indicator fixed, the sidelobe metrics are defined as
follows.

\begin{definition}[Three-Dimensional Sidelobe Metrics]\label{def:3d_sidelobe_metrics}
  For a realization $\mathbf{s}$, the instantaneous integrated sidelobe level (ISL), expected
  integrated sidelobe level (EISL), and expected sidelobe level (ESL) are
  \begin{align}
    \ISL_{3D}(\mathbf{s};\theta_0)
     & :=
    \int_{\Omega}\chi_{\rm sl}(k,q,\theta)
    \left|\mathcal{A}_{3D}(k,q;\theta,\theta_0)\right|^2
    \,\mathrm{d}\zeta,
    \label{eq:isl3d_inst} \\
    \EISL_{3D}(\theta_0)
     & :=
    \mathbb{E}\!\left[\ISL_{3D}(\mathbf{s};\theta_0)\right]
    \nonumber\\[-1pt]
     &=\int_{\Omega}\chi_{\rm sl}(k,q,\theta)\mathbb{E}\!\left[|\mathcal{A}_{3D}(k,q;\theta,\theta_0)|^2\right]\mathrm{d}\zeta,
    \label{eq:EISL_def} \\
    \ESL_{3D}(\theta_0)
     & :=
    \frac{\EISL_{3D}(\theta_0)}
    {\zeta\!\left(\Omega_{\rm sl}(\theta_0)\right)}.
    \label{eq:ESL_def}
  \end{align}
  The ESL is defined for a nonempty sidelobe region satisfying
  $0<\zeta(\Omega_{\rm sl}(\theta_0))<\infty$.
\end{definition}

\begin{remark}[Angular Sampling and Metric Comparisons]\label{rem:metric_sampling}
  When angle is sampled on a finite grid, the angular integrals are replaced by weighted sums and
  $\zeta(\Omega_{\rm sl})$ by the total sidelobe weight. Comparisons of EISL or ESL must use the
  same delay-Doppler domain, angular mainlobe boundary, angular weights, and normalization.
\end{remark}

\subsection{Multi-Symbol and FST Models}

We now extend the single-symbol model to coherent processing over $M$ OFDM blocks. The
multi-symbol derivation follows the fast-slow time (FST) route in~\cite{zhang2025discreteaf}:
cyclic-prefix (CP)
operations are made explicit first, and the Doppler phase is then separated into a slow-time
component across blocks and a fast-time component within each block.

Let $\mathbf{C}_{\rm cp}\in\CN^{(N+N_{\rm CP})\times N}$ insert the CP and let
$\mathbf{R}_{\rm cp}\in\CN^{N\times(N+N_{\rm CP})}$ remove it, with
$\mathbf{R}_{\rm cp}\mathbf{C}_{\rm cp}=\mathbf{I}_N$. For the $m$th OFDM block, let
$\mathbf{x}_m$ denote the CP-free time-domain transmit vector and define the serialized CPI vector
\begin{equation}
  \widetilde{\mathbf{x}}
  :=
  \begin{bmatrix}
    (\mathbf{C}_{\rm cp}\mathbf{x}_0)^T, \cdots,
    (\mathbf{C}_{\rm cp}\mathbf{x}_{M-1})^T
  \end{bmatrix}^{T}
  \label{eq:fst_serial_cpi}
\end{equation}
For a target $\ell$, let $q_{0,\ell}$ denote its integer Doppler index on the serialized
$M(N+N_{\rm CP})$-sample CPI grid. The Doppler phase at sample $(m,n)$, where
$m=0,\ldots,M-1$ and $n=0,\ldots,N+N_{\rm CP}-1$, factors as
\begin{equation}
 \begin{gathered}
 \resizebox{0.90\columnwidth}{!}{$\displaystyle
 e^{\frac{j2\pi q_{0,\ell}[m(N+N_{\rm CP})+n]}{M(N+N_{\rm CP})}}
 =e^{\frac{j2\pi q_{0,\ell}m}{M}}
 e^{\frac{j2\pi q_{0,\ell}n}{M(N+N_{\rm CP})}}$}\\[-1pt]
 \approx e^{\frac{j2\pi q_{0,\ell}m}{M}}.
 \end{gathered}
 \label{eq:fst_phase_factorization}
\end{equation}
For a given target $\ell$, define the dimensionless within-block Doppler-phase parameter
$\eta$. The approximation above holds under the FST condition
\begin{equation}
 \eta:=\frac{2\pi |q_{0,\ell}|N}{M(N+N_{\rm CP})}\ll1.
  \label{eq:fst_small_phase}
\end{equation}
The parameter $\eta$ measures the magnitude of the Doppler-phase change over $N$ sample
intervals within one OFDM block; the slow-time rotation across blocks is retained.
It is evaluated separately for each target, and the FST condition requires $\eta\ll1$
for every modeled target. 

For the serialized CPI grid, the slow-time Doppler bin is the residue of the
serialized index modulo the number of OFDM blocks, i.e.,
\begin{equation}
  q_\ell := [q_{0,\ell}]_{M}\in\mathbb{Z}_{M},
  \label{eq:q0_to_q_mapping}
\end{equation}
where $[\cdot]_M$ denotes the residue modulo $M$ (or its selected centered representative).
Consequently,
\begin{equation}
  e^{j2\pi q_{0,\ell}m/M}=e^{j2\pi q_\ell m/M},\qquad m=0,\ldots,M-1.
  \label{eq:slow_time_phase_mapping}
\end{equation}
Thus, $q_{0,\ell}$ and $q_\ell$ produce the same slow-time phase across the $M$ blocks,
whereas $q_{0,\ell}$ additionally determines the residual within-block phase
$e^{j2\pi q_{0,\ell}n/[M(N+N_{\rm CP})]}$ that is neglected under the FST condition.

Under \eqref{eq:fst_small_phase}, the fast-time phase in each CP-extended block is replaced by a
common block phase. The common phase associated with the CP position is absorbed into the complex
target coefficient $\beta_\ell$, and the serialized Doppler operator satisfies
\begin{equation}
 \mathbf{D}_{M(N+N_{\rm CP}),q_{0,\ell}} =
 \underbrace{\operatorname{diag}\!\left(
  1,\ldots,e^{j2\pi(M-1)q_\ell/M}
 \right)}_{\displaystyle \mathbf{D}_{M,q_\ell}}
 \otimes \mathbf{I}_{N+N_{\rm CP}}
\label{eq:fst_doppler_kron}
\end{equation}
where $q_\ell$ is the retained slow-time Doppler bin. Thus, the slow-time phase carried by block
$m$ is explicitly $e^{j2\pi m q_\ell/M}$. After CP removal, the blockwise received signal is
\begin{equation}
\mathbf{Y}_m =\sum_{\ell=1}^{L}\beta_\ell e^{j2\pi m q_\ell/M}\mathbf{H}_s(\theta_\ell)\mathbf{X}_m\mathbf{J}_{N,k_\ell}^{T}+\mathbf{Z}_m.
\label{eq:fst_block_echo}
\end{equation}
where $m=0,\ldots,M-1$. Here $\mathbf{Y}_m,\mathbf{Z}_m\in\CN^{N_r\times N}$, and the delay and angle parameters have
the same meanings as in the single-symbol model. In the FST subsection, $q\in\mathbb{Z}_M$
always denotes a slow-time Doppler bin; it is distinct from the single-symbol
intra-symbol Doppler index $q\in\mathbb{Z}_N$ introduced in Section~\ref{sec:system_model}.

For completeness, define the $m$th symbol data and transmit matrices by
\begin{equation}
  \begin{aligned}
  \mathbf{s}_m&:=
  \begin{bmatrix}\mathbf{s}_{m,0}^{T}, \cdots, \mathbf{s}_{m,N-1}^{T}\end{bmatrix}^{T},
  \\
  \mathbf{V}_m&:=
  \begin{bmatrix}\mathbf{W}_{m,0}\mathbf{s}_{m,0}, \cdots, \mathbf{W}_{m,N-1}\mathbf{s}_{m,N-1}\end{bmatrix},
  \\
  \mathbf{X}_m&=\mathbf{V}_m\mathbf{F}_N^{H}.
  \end{aligned}
  \label{eq:fst_symbol_matrices}
\end{equation}
Let
\begin{equation}
  \begin{aligned}
  \mathbf{G}_{T,m}(\theta)&:=
  \operatorname{blkdiag}_{n=0}^{N-1}\!\left\{\mathbf{a}_t^{T}(\theta)\mathbf{W}_{m,n}\right\},
  \\
  \mathbf{X}_m^{T}\mathbf{a}_t(\theta)&=\mathbf{F}_N^{H}\mathbf{G}_{T,m}(\theta)\mathbf{s}_m.
  \end{aligned}
  \label{eq:fst_x_theta}
\end{equation}
Stack the symbol vectors as
\begin{equation}
  \begin{aligned}
  \bar{\mathbf{s}}&:=
  \begin{bmatrix}\mathbf{s}_0^{T}, \cdots, \mathbf{s}_{M-1}^{T}\end{bmatrix}^{T}\in\CN^{MNN_s},\\[-1pt]
  \bar{\mathbf{G}}_T(\theta)&:=
  \operatorname{blkdiag}\!\left(\mathbf{G}_{T,0}(\theta),\ldots,\mathbf{G}_{T,M-1}(\theta)\right).
  \end{aligned}
  \label{eq:fst_stacked_symbols}
\end{equation}
The same i.i.d.\ symbol-statistics assumption used for the single-symbol analysis is imposed on
all entries of $\bar{\mathbf{s}}$.

\begin{definition}[FST Template and Matched-Filter Output]\label{def:fst_template_mf}
For a test cell $(k,q,\theta)$ with $k\in\mathcal{K}$ and $q\in\mathbb{Z}_M$, let
\begin{equation}
 \mathbf{T}_{{\rm p},m}^{\rm FST}(k,\theta)
 :=
 \mathbf{H}_s(\theta)\mathbf{X}_m\mathbf{J}_{N,k}^{T}.
 \label{eq:fst_template}
\end{equation}
The coherent CPI matched-filter output is
\begin{equation}
 z_{\rm MF}^{\rm FST}(k,q,\theta)
 =
  \sum_{m=0}^{M-1}e^{-j2\pi m q/M}
 \operatorname{Tr}\!\left[
 \mathbf{X}_m^{H}\mathbf{H}_s^{H}(\theta)\mathbf{Y}_m
 \mathbf{J}_{N,k}
 \right].
 \label{eq:fst_mf_def}
\end{equation}
\end{definition}

Using the block echo model, define the FST cross 3D-AF directly from the coherent single-target
contribution as follows.

\begin{definition}[FST Cross 3D Ambiguity Function]\label{def:fst_af}
For a test angle $\theta$ and target angle $\vartheta$, the FST cross 3D-AF is
\begin{align}
 \mathcal{A}_{3D}^{\rm FST}(k,q;\theta,\vartheta)
 &=\rho_r(\theta,\vartheta)
 \sum_{m=0}^{M-1}e^{-j2\pi mq/M}
 \nonumber\\[-1pt]
 &\quad\times
 \mathbf{a}_t^H(\theta)\mathbf{X}_m^*
 \mathbf{J}_{N,k}^{T}\mathbf{X}_m^T\mathbf{a}_t(\vartheta).
 \label{eq:fst_af_def}
\end{align}
The test Doppler $q\in\mathbb{Z}_M$ is a slow-time bin, whereas the single-symbol AF uses
$q\in\mathbb{Z}_N$ for the intra-symbol Doppler coordinate.
\end{definition}
Using this definition, the coherent matched-filter output in \eqref{eq:fst_mf_def} can be
written compactly as
\begin{equation}
 z_{\rm MF}^{\rm FST}(k,q,\theta)
 =\sum_{\ell=1}^{L}\beta_\ell
 \mathcal{A}_{3D}^{\rm FST}
 (k-k_\ell,q-q_\ell;\theta,\theta_\ell)
 +z_{k,q,\theta}^{\rm FST}.
 \label{eq:fst_shifted_af}
\end{equation}
For a target angle $\theta_0$, define the FST search region
\begin{equation}
 \Omega_{\rm FST}:=\mathcal{K}\times\mathbb{Z}_M\times\Theta.
 \label{eq:fst_search_region}
\end{equation}
With an angular mainlobe $\Theta_{\rm ml}(\theta_0)\subseteq\Theta$, the FST mainlobe and sidelobe
regions are
\begin{equation}
 \begin{aligned}
 \Omega_{{\rm FST},\rm ml}(\theta_0)
 :=
 \{(0,0,\theta):\theta\in\Theta_{\rm ml}(\theta_0)\},\\[-1pt]
 \Omega_{{\rm FST},\rm sl}(\theta_0)
 :=
 \Omega_{\rm FST}\setminus\Omega_{{\rm FST},\rm ml}(\theta_0).
 \end{aligned}
 \label{eq:fst_sidelobe_region}
\end{equation}
The product measure contains counting measures on $k$ and $q$ and Lebesgue measure on angle:
\begin{equation}
 \int_{\Omega_{\rm FST}} f\,\mathrm{d}\zeta_{\rm FST}
 :=
 \sum_{k\in\mathcal{K}}\sum_{q\in\mathbb{Z}_M}
 \int_{\Theta} f(k,q,\theta)\,\mathrm{d}\theta.
 \label{eq:fst_product_measure}
\end{equation}

\begin{definition}[FST Sidelobe Metrics]\label{def:fst_sidelobe_metrics}
The instantaneous integrated sidelobe level (ISL), expected integrated sidelobe level (EISL), and
expected sidelobe level (ESL) are
\begin{align}
 \ISL_{3D}^{\rm FST}(\bar{\mathbf{s}};\theta_0)
 &:=
 \int_{\Omega_{{\rm FST},\rm sl}(\theta_0)}
 \left|\mathcal{A}_{3D}^{\rm FST}(k,q;\theta,\theta_0)\right|^2
 \,\mathrm{d}\zeta_{\rm FST}, \\
 \EISL_{3D}^{\rm FST}(\theta_0)
 &:=
 \mathbb{E}\!\left[\ISL_{3D}^{\rm FST}(\bar{\mathbf{s}};\theta_0)\right], \\
 \ESL_{3D}^{\rm FST}(\theta_0)
 &:=
 \frac{\EISL_{3D}^{\rm FST}(\theta_0)}
 {\zeta_{\rm FST}\!\left(\Omega_{{\rm FST},\rm sl}(\theta_0)\right)}.
 \label{eq:fst_metrics}
\end{align}
The ESL is defined when the sidelobe region has finite nonzero measure.
\end{definition}

\section{Statistical Characterization of the 3D-AF}\label{sec:statistical_analysis}

Building on the matched-filter and quadratic-form representations established in
Section~\ref{sec:system_model}, we now characterize the random 3D-AF statistically. We first examine
the deterministic block structure of the single-symbol kernel and then derive its second-order
statistics before extending the analysis to coherent multi-symbol processing and the FST
approximation.

\subsection{Single-Symbol Quadratic-Form Structure}

The matrix $\mathbf{M}_{k,q}$ separates the random symbols from the deterministic sensing and
transmission parameters: $\boldsymbol{\Psi}_{k,q}$ contains the delay-Doppler dependence, whereas
$\rho_r$ and $\mathbf{G}_T$ contain the array and precoding dependence.

\begin{equation}
  \begin{gathered}
    \mathbf{M}_{k,q}(\theta,\vartheta)
    =
    \begin{bmatrix}
      \mathbf{M}_{0,0}^{(k,q)}   & \cdots &
      \mathbf{M}_{0,N-1}^{(k,q)}                   \\
      \vdots                     & \ddots & \vdots \\
      \mathbf{M}_{N-1,0}^{(k,q)} & \cdots &
      \mathbf{M}_{N-1,N-1}^{(k,q)}
    \end{bmatrix}
    \in\CN^{NN_s\times NN_s},\\[-1pt]
    \mathbf{M}_{n,n'}^{(k,q)}(\theta,\vartheta)
    =
    [\mathbf{M}_{k,q}(\theta,\vartheta)]_{n,n'}
    \in\CN^{N_s\times N_s}.
  \end{gathered}
  \label{eq:M_block_partition}
\end{equation}

\begin{proposition}[Subcarrier Block Structure]\label{prop:block_structure}
  Because $\mathbf{G}_T(\theta)$ is block diagonal, $\mathbf{M}_{k,q}(\theta,\vartheta)$ has the
  $N\times N$ block partition in \eqref{eq:M_block_partition}. Its $(n,n')$th block is
  \begin{equation}
    \mathbf{M}_{n,n'}^{(k,q)}(\theta,\vartheta)=\rho_r(\theta,\vartheta)[\boldsymbol{\Psi}_{k,q}]_{n,n'}\mathbf{W}_n^H\mathbf{a}_t^*(\theta)\mathbf{a}_t^T(\vartheta)\mathbf{W}_{n'},
    \label{eq:block_structure}
  \end{equation}
  {On the integer cyclic grid, $[\boldsymbol{\Psi}_{k,q}]_{n,n'}
  =e^{j2\pi nk/N}$ if $n'=n+q\pmod N$, and is zero otherwise.
  Thus, the Doppler mismatch selects a pair of subcarriers, while the delay sets its phase.
  The receive-array correlation and the two precoders weight the angular responses of that pair.
  Each $\mathbf{s}_n^H\mathbf{M}_{n,n'}^{(k,q)}\mathbf{s}_{n'}$ is a contribution to the
  matched-filter 3D-AF; $\mathbf{M}_{n,n'}^{(k,q)}$ is its deterministic coupling weight.}
  The result follows by extracting the $(n,n')$th block of \eqref{eq:M_def}.
\end{proposition}

\subsection{Expected Squared 3D-AF and Sidelobe Metrics}

The quadratic-form representation reduces the statistical analysis of
the 3D-AF to the second moment of a quadratic form in the random symbol
vector. We first establish the required covariance identity and then
derive the expected squared 3D-AF.

\begin{lemma}[Quadratic-Form Covariance]\label{lem:quadratic_covariance}
  For deterministic matrices $\mathbf{P}$ and $\mathbf{Q}$ of compatible size,
  Assumption~\ref{ass:symbol_statistics} gives the quadratic-form covariance
  \begin{equation}
    \resizebox{0.98\columnwidth}{!}{$\displaystyle
      \operatorname{Cov}\!\left(\mathbf{s}^H\mathbf{P}\mathbf{s},
      \mathbf{s}^H\mathbf{Q}\mathbf{s}\right)
      =\operatorname{Tr}\!\left(\mathbf{P}\mathbf{Q}^H\right)
      +(\kappa-2)\sum_i [\mathbf{P}]_{i,i}[\mathbf{Q}]_{i,i}^*$}.
    \label{eq:lemma_quadratic_covariance}
  \end{equation}
  \emph{Proof:} See Appendix~\ref{app:covariance}.\par\smallskip
\end{lemma}

Under Assumption~\ref{ass:symbol_statistics},
$\mathbb{E}[\mathbf{s}^H\mathbf{P}\mathbf{s}]
  =\operatorname{Tr}(\mathbf{P})$. Hence, the desired second moment is
obtained by adding the squared magnitude of this mean to the covariance
in Lemma~\ref{lem:quadratic_covariance}. Applying this relation to the
quadratic-form representation in \eqref{eq:quadratic_af} gives the
following result.

For compactness, in the following theorem we use
$\mathbf{M}$ to denote $\mathbf{M}_{k,q}(\theta,\vartheta)$.

\begin{theorem}[Expected Squared 3D Ambiguity Function]\label{thm:expected_squared_af}
  Under the system model of Section~\ref{sec:system_model},
  Assumption~\ref{ass:symbol_statistics}, and the quadratic-form representation in
  \eqref{eq:quadratic_af}, the expected squared cross 3D-AF is
  \begin{equation}
\begin{aligned}
% 第一行：仅放左侧长表达式，& 作为全局对齐基准
&\mathbb{E}\!\left[|\mathcal{A}_{3D}(k,q;\theta,\vartheta)|^2\right] \\
% 第二行：第一个等号，和首行左对齐
&= \operatorname{Cov}\!\left(\mathbf{s}^H\mathbf{M}\mathbf{s}, \mathbf{s}^H\mathbf{M}\mathbf{s}\right) +\left|\operatorname{tr}(\mathbf{M})\right|^2 \\
% 第三行：第二个等号，所有项可完整放在一行
&= \left|\operatorname{tr}(\mathbf{M})\right|^2+\operatorname{tr}(\mathbf{M}\mathbf{M}^H) + (\kappa-2)\sum_i|[\mathbf{M}]_{i,i}|^2.
\end{aligned}
\label{eq:expected_squared_af}
\end{equation}
  \emph{Proof:} Apply Lemma~\ref{lem:quadratic_covariance} with
  $\mathbf{P}=\mathbf{Q}=\mathbf{M}$. The mean
  identity above yields the first equality, and
  \eqref{eq:quadratic_af} completes the proof. A full fourth-order
  derivation is given in Appendix~\ref{app:covariance}.\par\smallskip
\end{theorem}

{
For fixed $(k,q,\theta,\vartheta)$, let $\mathbf{M}_{n,n'}$ denote the corresponding
subblock in \eqref{eq:M_block_partition}, with $n,n'\in\mathbb{Z}_N$.
The blockwise form of \eqref{eq:expected_squared_af} is
\begin{equation}
  \begin{aligned}
  &\mathbb{E}\!\left[|\mathcal{A}_{3D}(k,q;\theta,\vartheta)|^2\right]\\
  &=\left|\sum_n\tr(\mathbf{M}_{n,n})\right|^2
    +\sum_{n,n'}\|\mathbf{M}_{n,n'}\|_F^2\\[-1pt]
  &\quad +(\kappa-2)\sum_n\sum_{r=1}^{N_s}
    |[\mathbf{M}_{n,n}]_{r,r}|^2.
  \end{aligned}
  \label{eq:expected_squared_af_blocks}
\end{equation}
The coherent and kurtosis-dependent terms involve only diagonal subcarrier blocks, whereas
the kernel-energy term includes all blocks. Thus, the three terms can be accumulated without
forming the full $NN_s\times NN_s$ matrix.
}

\begin{remark}[Statistical Decomposition and Design Implications]\label{rem:statistical_decomposition}
  The three terms in \eqref{eq:expected_squared_af} have distinct
  interpretations. The coherent term $|\operatorname{Tr}(\mathbf{M})|^2$
  is the squared magnitude of the mean quadratic form. The term
  $\|\mathbf{M}\|_F^2$ measures the total kernel energy, whereas the
  diagonal correction is the only term controlled by the constellation
  kurtosis $\kappa$. For proper Gaussian symbols this correction
  vanishes, whereas for proper QAM and constant-modulus PSK it is
  negative because $\kappa<2$.

  The formula also identifies two complementary design mechanisms. Spatial or frequency-selective
  precoding changes the full matrix $\mathbf{M}$ and can therefore reshape all three terms.
  Constellation shaping changes $\kappa$ and acts only through the diagonal energy. This distinction
  relates to earlier random-signal precoding~\cite{lu2024randomprecoding} and
  constellation-shaping and input-distribution studies~\cite{du2024reshaping,inputdist2024},
  respectively, while \eqref{eq:expected_squared_af} makes the
  interaction explicit over a three-dimensional search domain.
\end{remark}

\begin{corollary}[Deterministic Evaluation of 3D EISL and ESL]\label{cor:eisl_esl}
  Under the conditions of Theorem~\ref{thm:expected_squared_af}, the EISL in
  Definition~\ref{def:3d_sidelobe_metrics} is given below. For each integration point
  $(k,q,\theta)$, we write $\mathbf{M}:=\mathbf{M}_{k,q}(\theta,\theta_0)$ for compactness:
\begin{equation}
  \begin{aligned}
    \EISL_{3D}(\theta_0)
     & =
    \int_{\Omega}\chi_{\rm sl}(k,q,\theta)
    \Bigl[ \quad
    |\tr(\mathbf{M})|^2
    +\|\mathbf{M}\|_F^2\\[-1pt]
     & \quad+
    (\kappa-2)\sum_i
    |[\mathbf{M}]_{i,i}|^2
    \Bigr]\,\mathrm{d}\zeta.
    \end{aligned}
    \label{eq:EISL_closed_form}
  \end{equation}
  The ESL follows by dividing \eqref{eq:EISL_closed_form} by
  $\zeta(\Omega_{\rm sl}(\theta_0))$, as specified in \eqref{eq:ESL_def}.

  \emph{Proof:} Substitute \eqref{eq:expected_squared_af} into
  \eqref{eq:EISL_def} and use linearity of expectation. Details are given
  in Appendix~\ref{app:covariance}.\par\smallskip
\end{corollary}
\subsection{FST Statistical Characterization}

Having obtained the FST matched-filter output and cross 3D-AF, we now separate the random
symbol vector from the deterministic sensing operators. Substituting
\eqref{eq:fst_x_theta} into \eqref{eq:fst_af_def} gives the following representation.
For compactness, in the following FST expressions we write
$\mathbf{M}^{\rm FST}$ for
$\bar{\mathbf{M}}_{k,q}^{\rm FST}(\theta,\vartheta)$.
\begin{equation}
 \mathcal{A}_{3D}^{\rm FST}(k,q;\theta,\vartheta)
 =
 \bar{\mathbf{s}}^{H}
 \mathbf{M}^{\rm FST}
 \bar{\mathbf{s}}.
 \label{eq:fst_quadratic}
\end{equation}
The delay kernel that appears in this expression is
\begin{equation}
 \boldsymbol{\Psi}_{k,0}
 :=
 \mathbf{F}_N\mathbf{J}_{N,k}^{T}\mathbf{F}_N^{H}
 \in\CN^{N\times N}.
 \label{eq:fst_psi}
\end{equation}

\begin{proposition}[FST Quadratic-Form Matrix]\label{prop:fst_matrix}
For $q\in\mathbb{Z}_M$, the deterministic matrix in \eqref{eq:fst_quadratic} is
\begin{equation}
 \mathbf{M}^{\rm FST}
 =
 \rho_r(\theta,\vartheta)
 \bar{\mathbf{G}}_T^{H}(\theta)
 \left(\mathbf{D}_{M,q}^{*}\otimes\boldsymbol{\Psi}_{k,0}\right)
 \bar{\mathbf{G}}_T(\vartheta).
 \label{eq:fst_matrix}
\end{equation}
\end{proposition}

\begin{theorem}[Expected Squared FST 3D-AF]\label{thm:fst_expected}
Under the i.i.d.\ symbol model, applying the quadratic-form moment identity used to derive
\eqref{eq:expected_squared_af} to \eqref{eq:fst_quadratic} yields the expected squared FST 3D-AF:
\begin{equation}
 \begin{aligned}
 \mathbb{E}\!\left[|\mathcal{A}_{3D}^{\rm FST}|^2\right]
 &=\left|\tr(\mathbf{M}^{\rm FST})\right|^2
   +\|\mathbf{M}^{\rm FST}\|_F^2\\[-1pt]
 &\quad+(\kappa-2)\sum_{i=1}^{MNN_s}
   |[\mathbf{M}^{\rm FST}]_{i,i}|^2.
 \end{aligned}
 \label{eq:fst_expected}
\end{equation}
\end{theorem}

\begin{corollary}[Deterministic FST EISL and ESL]\label{cor:fst_eisl}
For a target at angle $\theta_0$, setting $\vartheta=\theta_0$ in
\eqref{eq:fst_expected} and integrating over the FST sidelobe region gives the EISL:
\begin{equation}
 \begin{aligned}
 \EISL_{3D}^{\rm FST}(\theta_0)
 &=\int_{\Omega_{{\rm FST},\rm sl}(\theta_0)}\biggl[
 |\tr(\mathbf{M}^{\rm FST})|^2+\|\mathbf{M}^{\rm FST}\|_F^2\\[-1pt]
 &\quad +(\kappa-2)\sum_{i=1}^{MNN_s}|[\mathbf{M}^{\rm FST}]_{i,i}|^2
 \biggr]\mathrm{d}\zeta_{\rm FST}.
 \end{aligned}
 \label{eq:fst_eisl_closed}
\end{equation}
The ESL is obtained by dividing \eqref{eq:fst_eisl_closed} by
$\zeta_{\rm FST}(\Omega_{{\rm FST},\rm sl}(\theta_0))$.
\end{corollary}

\begin{corollary}[Kronecker Structure With Identical Precoders]\label{cor:fst_kron}
If $\mathbf{W}_{m,n}=\mathbf{W}_n$ for all $m$ and $n$, then
\begin{equation}
 \mathbf{M}^{\rm FST}
 =
 \mathbf{D}_{M,q}^{*}\otimes\mathbf{M}_{k,0}(\theta,\vartheta).
 \label{eq:fst_kron}
\end{equation}
Consequently, the FST 3D-AF has the equivalent form
\begin{equation}
 \mathcal{A}_{3D}^{\rm FST}(k,q;\theta,\vartheta)
 =
 \bar{\mathbf{s}}^{H}\mathbf{M}^{\rm FST}\bar{\mathbf{s}}
 =
 \bar{\mathbf{s}}^{H}
 \left[\mathbf{D}_{M,q}^{*}\otimes\mathbf{M}_{k,0}(\theta,\vartheta)\right]
 \bar{\mathbf{s}}.
 \label{eq:fst_af_kron}
\end{equation}
\end{corollary}

\section{Numerical Results and Discussion}\label{sec:numerical_results}

In this section, numerical results are presented to validate the preceding theoretical derivations
and approximations. Unless otherwise stated, we use $N=32$ subcarriers, $N_t=N_r=8$
half-wavelength-spaced ULA elements, one data stream per subcarrier, a boresight target at
$\theta_0=0^\circ$, and unit-power 16-QAM symbols with fourth moment $\kappa=1.32$. The default
precoder is $\mathbf{W}_n=\mathbf{a}_t^*(\theta_0)$ on every subcarrier. For all multi-symbol/FST
numerical experiments, we use $M=32$ OFDM blocks. The delay and Doppler sets
contain all $N$ cyclic bins displayed through centered representatives, and the numerical angle
samples cover $[-60^\circ,60^\circ]$ in $2^\circ$ increments to approximate the continuous angular
integral. The angular mainlobe is defined by the first nulls of the receive ULA, corresponding to approximately \([-14.48^\circ,14.48^\circ]\) for the boresight target. Monte Carlo results use $R=5000$ independent payload realizations.
Theoretical curves are evaluated from the derived quadratic-form expectation, while numerical
curves are obtained by averaging over the same grids and normalization.

For each discrete delay-Doppler cell and sampled angle $\omega=(k,q,\theta)$, let
$\mathcal{A}_{3D}^{(r)}(\omega)$ denote the AF generated by the $r$th independent payload
realization. With $R$ trials, the Monte Carlo estimate
of the expected squared AF is
\begin{equation}
  \widehat P_R(\omega)
  :=
  \frac{1}{R}
  \sum_{r=1}^{R}
  \left|\mathcal{A}_{3D}^{(r)}(k,q;\theta,\theta_0)\right|^2.
  \label{eq:mc_power_estimator}
\end{equation}
This estimator is compared with the closed-form expectation in \eqref{eq:expected_squared_af}. Each
trial regenerates the random payload while keeping the precoders, target angle, grids, and
normalization fixed. Theory and simulation use the same mainlobe exclusion rule. As specified
below, the reported sampled EISL uses unit counting weights on the common angle samples; radian
quadrature weights may instead be used to approximate the continuous-angle metric in
\eqref{eq:EISL_def}.

\subsection{Joint Delay-Doppler-Angle Structure}

We first examine the single-symbol point-target AF to identify how delay, Doppler, and
angle mismatches shape its expected response. This geometric characterization provides
the reference for the constellation comparisons below. To compare the conditional slices,
we normalize the expected squared AF by its value at the fully matched cell:
\begin{equation}
  \begin{aligned}
    P_{\rm th}(k,q,\theta)
     & =
    \frac{\mathbb{E}[|\mathcal{A}_{3D}(k,q;\theta,\theta_0)|^2]}
    {\mathbb{E}[|\mathcal{A}_{3D}(0,0;\theta_0,\theta_0)|^2]}.
  \end{aligned}
  \label{eq:numerical_normalized_af}
\end{equation}
The matched cell $(0,0,\theta_0)$ is therefore $0$ dB. The centered delay and Doppler
indices are expressed in physical units as $\tau=kT_s$ and $f_d=q/(NT_s)$, with
$T_s=31.25\,{\rm ns}$.

\begin{figure}[!htbp]
  \centering
  \includegraphics[width=0.98\linewidth]{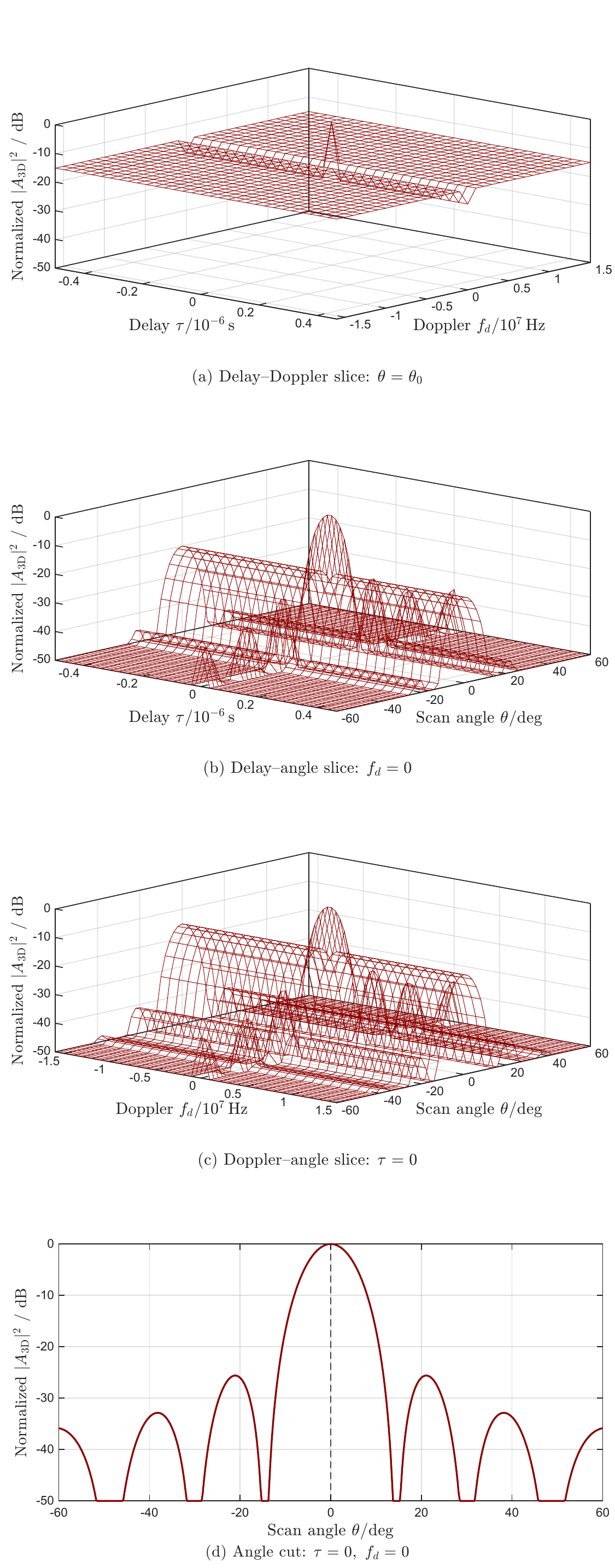}
  \caption{Normalized expected squared 3D-AF slices in decibels: (a) delay-Doppler slice at
    $\theta=\theta_0$, (b) delay-angle slice at $q=0$, (c) Doppler-angle slice at $k=0$, and (d) angular cut at $k=q=0$.}
  \label{fig:joint_3d_af_slices}
\end{figure}

With the test angle matched to the target, the delay-Doppler slice in
Fig.~\ref{fig:joint_3d_af_slices}(a) exhibits a central peak above an expected sidelobe
floor. This floor is governed jointly by the cyclic delay-Doppler kernel and the symbol
statistics in \eqref{eq:expected_squared_af}. The delay-angle and Doppler-angle slices in
panels (b) and (c) combine the respective grid mismatch with receive-array correlation and
the transmit precoder projection. At $k=q=0$, panel (d) isolates the angular response of
the array and precoder. Its continuous variation contrasts with the sampled integer delay
and Doppler axes. Because each panel fixes one or two coordinates, its sidelobes describe
conditional ambiguity and do not by themselves determine the largest response over the
full 3D search region. We therefore examine the aggregate sidelobe energy next, followed
by coordinate cuts that identify where constellation-dependent differences arise.

\subsection{Constellation Effects in Single-Symbol and FST Processing}

We compare unit-power QPSK, 16-QAM, 256-QAM, and proper Gaussian symbols under OFDM
signaling, first for a single symbol and then for coherent FST processing. The precoders,
array configuration, angular samples, and mainlobe boundary are held fixed so that the
symbol distribution is the comparison variable. Let $\Theta^{\rm samp}=\{\theta_i\}$ and
let $w_i$ denote its prescribed angular weights. The sampled sums below describe
aggregation over the 3D search domain; constellation comparisons use common weights
and power scaling within each panel.
For a sampled sidelobe region $\Omega_{\rm sl}^{\rm
    samp}\subseteq\mathcal{K}\times\mathcal{Q}\times\Theta^{\rm samp}$, the sampled full 3D EISL is
\begin{equation}
  \begin{aligned}
    \EISL_{\rm samp}
     & =
    \sum_{(k,q,\theta_i)\in\Omega_{\rm sl}^{\rm samp}}
    w_i
    \mathbb{E}\left[
      |\mathcal{A}_{3D}(k,q;\theta_i,\theta_0)|^2
      \right],
  \end{aligned}
  \label{eq:sampled_eisl}
\end{equation}
where $\Omega_{\rm sl}^{\rm samp}$ excludes only
$\{(0,0,\theta_i):\theta_i\in\Theta_{\rm ml}(\theta_0)\}$, as in
\eqref{eq:sidelobe_region}. Cells with $(k,q)\ne(0,0)$ remain in the sidelobe sum even when
$\theta_i\in\Theta_{\rm ml}(\theta_0)$. We use $w_i=1$ for the sampled sum;
radian quadrature weights would instead approximate the continuous-angle
metric in \eqref{eq:EISL_def}.

For FST processing, the corresponding sampled metric is
\begin{equation}
  \begin{aligned}
    \EISL_{\rm samp}^{\rm FST}
    &=\sum_{(k,q,\theta_i)\in\Omega_{{\rm FST},\rm sl}^{\rm samp}}
      w_i\mathbb{E}\!\left[
      |\mathcal{A}_{3D}^{\rm FST}(k,q;\theta_i,\theta_0)|^2
      \right].
  \end{aligned}
  \label{eq:sampled_fst_eisl}
\end{equation}
Here, $k$ and $q$ are delay and Doppler mismatch coordinates. The FST sum uses all $M$
slow-time Doppler bins and excludes the matched delay-Doppler cell within the same
angular mainlobe, following \eqref{eq:fst_sidelobe_region}. It also uses $w_i=1$.
Theoretical values are obtained from \eqref{eq:expected_squared_af} for the
single-symbol case and \eqref{eq:fst_expected} for the FST case, with the latter
summed as in Corollary~\ref{cor:fst_eisl}. Monte Carlo values average the corresponding
ISL over independent payload realizations. A common positive power-normalization
factor changes the EISL level but preserves the constellation ordering at fixed $N$.

\begin{figure}[!htbp]
  \centering
  \includegraphics[width=0.98\linewidth]{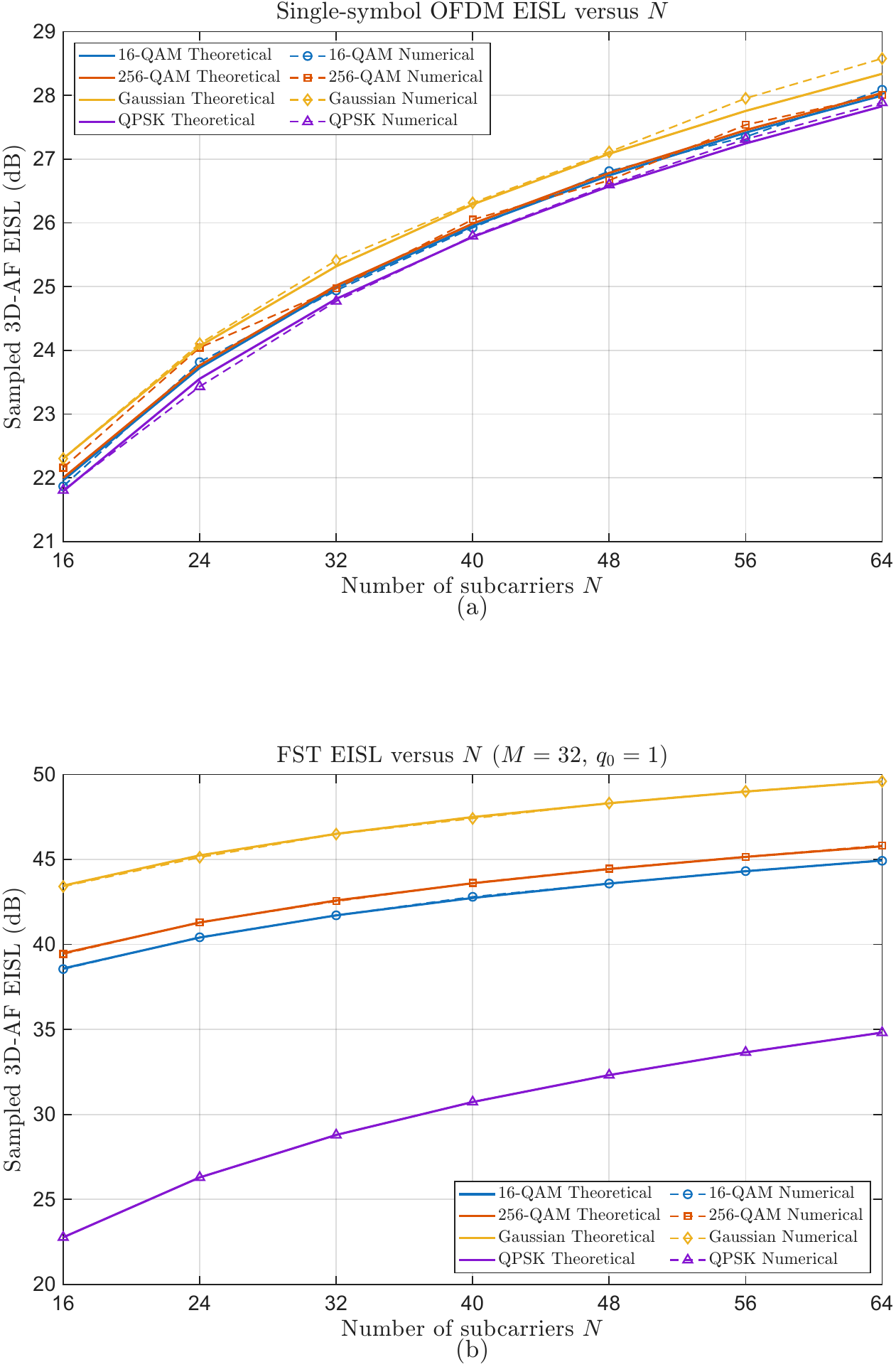}
  \caption{Constellation effects on sampled full 3D-AF EISL versus the number of subcarriers
    $N$, displayed in decibels: (a) single-symbol OFDM and (b) FST processing over
    $M=32$ OFDM blocks with single-target Doppler index $q_0=1$.}
  \label{fig:eisl_vs_subcarriers}
\end{figure}

\begin{figure}[!htbp]
  \centering
  \includegraphics[width=0.98\linewidth]{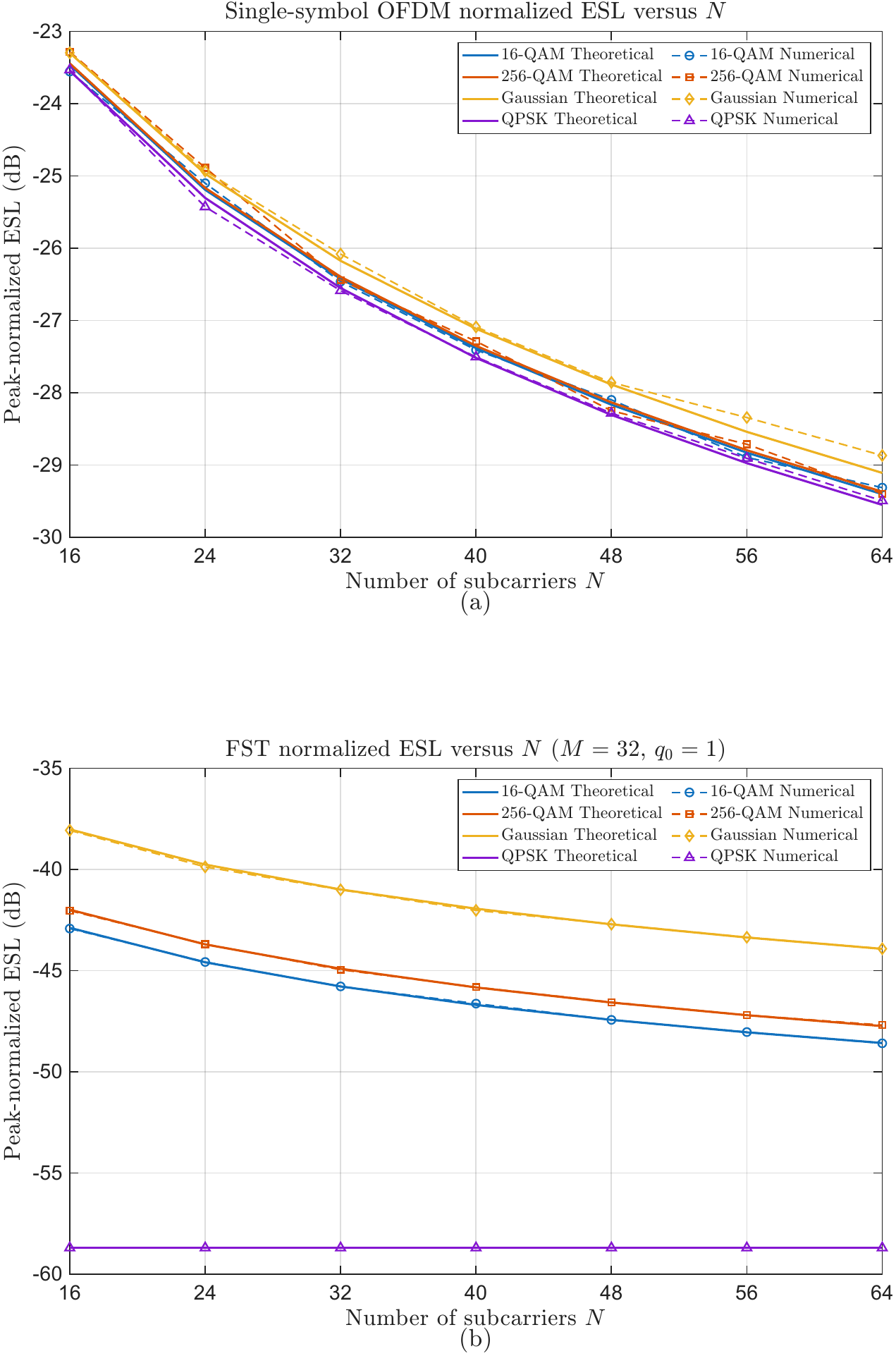}
  \caption{Peak-normalized expected sidelobe level (ESL) versus the number of subcarriers
  $N$, with theoretical curves and Monte Carlo markers: (a) single-symbol OFDM and
  (b) FST processing over $M=32$ OFDM blocks with target Doppler index $q_0=1$.}
  \label{fig:esl_vs_subcarriers}
\end{figure}

Figure~\ref{fig:eisl_vs_subcarriers} varies $N$ from $16$ to $64$ in steps of $8$.
Panel (a) uses the single-symbol delay-Doppler grid with $q\in\mathbb{Z}_N$.
Panel (b) uses $M=32$ OFDM blocks and a single target with Doppler index
$q_0=1$; its retained slow-time bin is $\operatorname{mod}(q_0,M)=1$, and the search
is expressed in relative bins $q\in\mathbb{Z}_M$.

In Fig.~\ref{fig:eisl_vs_subcarriers}(a), the numerical curves follow the overall
theoretical trends, although deviations are visible at some subcarrier
counts. At fixed $N$, the theoretical EISL increases from QPSK through 16-QAM and
256-QAM to proper Gaussian symbols. With unit symbol power and fixed precoders, this
ordering follows the increasing fourth moment $\kappa$ and the nonnegative diagonal
energy multiplying $(\kappa-2)$ in \eqref{eq:expected_squared_af}.

The FST results in Fig.~\ref{fig:eisl_vs_subcarriers}(b) retain the same theoretical
constellation ordering, and the numerical curves closely track the corresponding
predictions. The observed constellation dependence is consistent with the FST
statistical characterization. The fourth-moment dependence persists because
\eqref{eq:fst_expected} has the same statistical decomposition, with
$\mathbf{M}^{\rm FST}$ replacing the single-symbol kernel. The accuracy of the FST
approximation to the serialized Doppler response is examined separately in the next
subsection.

\begin{figure}[!htbp]
  \centering
  \includegraphics[width=1.0\linewidth]{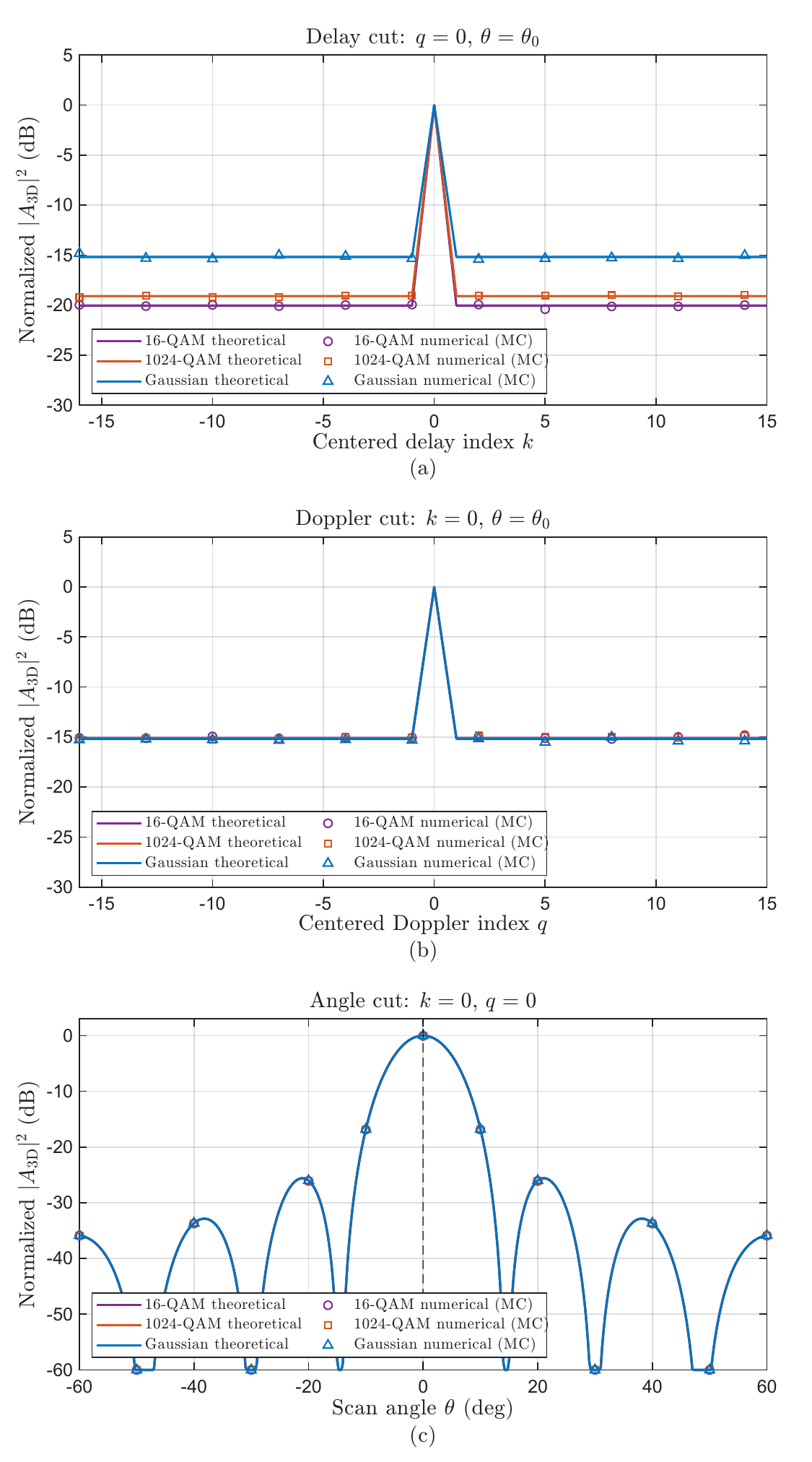}
  \caption{Normalized expected squared OFDM 3D-AF for 16-QAM, 1024-QAM, and proper Gaussian
    symbols, with theoretical curves and Monte Carlo markers: (a) delay cut at $q=0$ and
    $\theta=\theta_0$; (b) Doppler cut at $k=0$ and $\theta=\theta_0$; (c) angle cut at $k=q=0$
    for the same constellations and normalization.}
  \label{fig:constellation_af_cuts}
\end{figure}

Both panels show increasing aggregate EISL with $N$. Their absolute levels should be
interpreted with their respective processing intervals and Doppler grids: the
single-symbol metric sums $N$ Doppler bins, whereas the FST metric sums $M=32$
slow-time bins and uses coherent processing over 32 blocks. We therefore compare
constellations within each panel; interpreting the absolute level difference as a
sensing gain would additionally require a common energy and grid normalization.

Figure~\ref{fig:esl_vs_subcarriers} reports the corresponding peak-normalized ESL,
which divides the average sidelobe power by the expected matched-cell power within
the same processing model. This normalization removes the dominant change in total
mainlobe energy as $N$ varies and gives an average sidelobe-level comparison. In both
the single-symbol and FST panels, the theoretical curves and Monte Carlo markers are
closely aligned. The constant-modulus QPSK curve is generally the lowest, the proper
Gaussian curve is the highest, and the QAM curves lie between them, consistent with
the fourth-moment dependence. The single-symbol and FST ESL values are compared only
within their respective panels because their search grids and coherent processing
intervals differ.

To locate the constellation-dependent differences within the ambiguity response,
Fig.~\ref{fig:constellation_af_cuts} returns to the single-symbol case and shows
normalized delay, Doppler, and angle cuts for 16-QAM, 1024-QAM, and proper Gaussian
symbols. These cuts use the matched-cell normalization in
\eqref{eq:numerical_normalized_af} and resolve the response at individual cells,
complementing the aggregate EISL comparison in Fig.~\ref{fig:eisl_vs_subcarriers}.
The 1024-QAM case provides an additional higher-order constellation comparison for the local response.

Theory and Monte Carlo agree across the three displayed cuts. Along delay, the normalized
sidelobes increase from 16-QAM to 1024-QAM and then to proper Gaussian symbols with
$\kappa=2$, in line with the fourth-order correction in \eqref{eq:expected_squared_af}.
By comparison, the Doppler and angle cuts nearly overlap at the plotted scale. Under the
present precoder and normalization, the angular shape is primarily set by the
array-precoder response; the observed overlap does not imply constellation-independent
sidelobes for other configurations. Together, the EISL and cut comparisons distinguish
aggregate constellation effects from their distribution across the delay, Doppler, and
angle coordinates.

\subsection{FST Approximation Accuracy}

We assess the FST approximation using a single zero-delay boresight target with unit-power
16-QAM symbols, a CP length of $N_{\rm CP}=6$, and 100 Monte Carlo trials. All remaining
system parameters follow the defaults stated at the beginning of
Section~\ref{sec:numerical_results}. In each trial, the serialized
and FST AFs use the same payload. We vary the target Doppler index to examine how the
approximation error changes with $\eta$ defined in \eqref{eq:fst_small_phase}.

The CP length is selected in a separate illustrative simulation with an assumed four-tap
Rayleigh channel having delays $[0,2,4,6]$ samples and average powers $[0,-3,-6,-9]$ dB.
At the adopted sampling interval, its maximum excess delay is $6T_s=187.5$ ns.
Across 100 random channel and OFDM-block realizations, $N_{\rm CP}=6$ is the shortest
of $0{:}6$ meeting a worst-case linear-versus-circular-convolution relative-error
threshold of $10^{-10}$. This channel is used only to select the CP; it is not applied
to the AF comparison, which uses the CP-removed cyclic model. The CP affects the AF
comparison only through the block duration in the serialized Doppler phase and in $\eta$.

Since the AF experiment contains only one target, we omit the target index $\ell$ and
write $q_0$ for its integer Doppler index on the $M(N+N_{\rm CP})$-sample CPI grid.
Its retained slow-time Doppler bin is $q_{\rm t}:=\operatorname{mod}(q_0,M)$,
whereas $q$ denotes the test slow-time bin. With target delay $k_0=0$, the plotted
centered integer mismatches are
\begin{equation}
 \Delta k := (k-k_0)_{\!N},\qquad
 \Delta q := (q-q_{\rm t})_{\!M},
  \label{eq:fst_centered_mismatch}
\end{equation}
where the subscripts denote centered representatives modulo $N$ and $M$. Thus, $\Delta q$
is a slow-time bin mismatch, whereas $q_0$ determines the discarded within-block
phase. The matched cell is $(\Delta k,\Delta q)=(0,0)$.

With the test and target angles both fixed at $\theta_0$, denote the two responses by
$\mathcal{A}_{3D}^{\rm ser}$ and $\mathcal{A}_{3D}^{\rm FST}$, respectively. The pointwise
relative error is
\begin{equation}
 \varepsilon_{\rm AF}(\Delta k,\Delta q;q_0)
 :=10\log_{10}
 \frac{\mathbb{E}\!\left[
   \left|\mathcal{A}_{3D}^{\rm ser}
       -\mathcal{A}_{3D}^{\rm FST}\right|^2\right]}
 {\mathbb{E}\!\left[\left|\mathcal{A}_{3D}^{\rm ser}\right|^2\right]}.
 \label{eq:fst_af_approx_error}
\end{equation}
The numerator is the mean-square difference between the two AFs evaluated with the same
payload, rather than the difference between their expected powers.

For $q_0=1$, $4$, and $16$, the parameter $\eta$ defined in
\eqref{eq:fst_small_phase} is approximately $0.165$, $0.661$, and $2.646$,
respectively. Under otherwise identical settings, the relative AF discrepancy in
Fig.~\ref{fig:fst_approximation_error} increases overall with $\eta$, consistent
with the growing within-block Doppler phase discarded by FST\@.
These noiseless, integer-grid results characterize pointwise AF approximation error
and complement the statistical EISL comparison in
Fig.~\ref{fig:eisl_vs_subcarriers}(b): the latter assesses constellation-dependent
statistics within the FST model, while the present experiment evaluates the error
introduced by the FST Doppler approximation.

\begin{figure}[!htbp]
  \centering
  \includegraphics[width=0.99\linewidth]{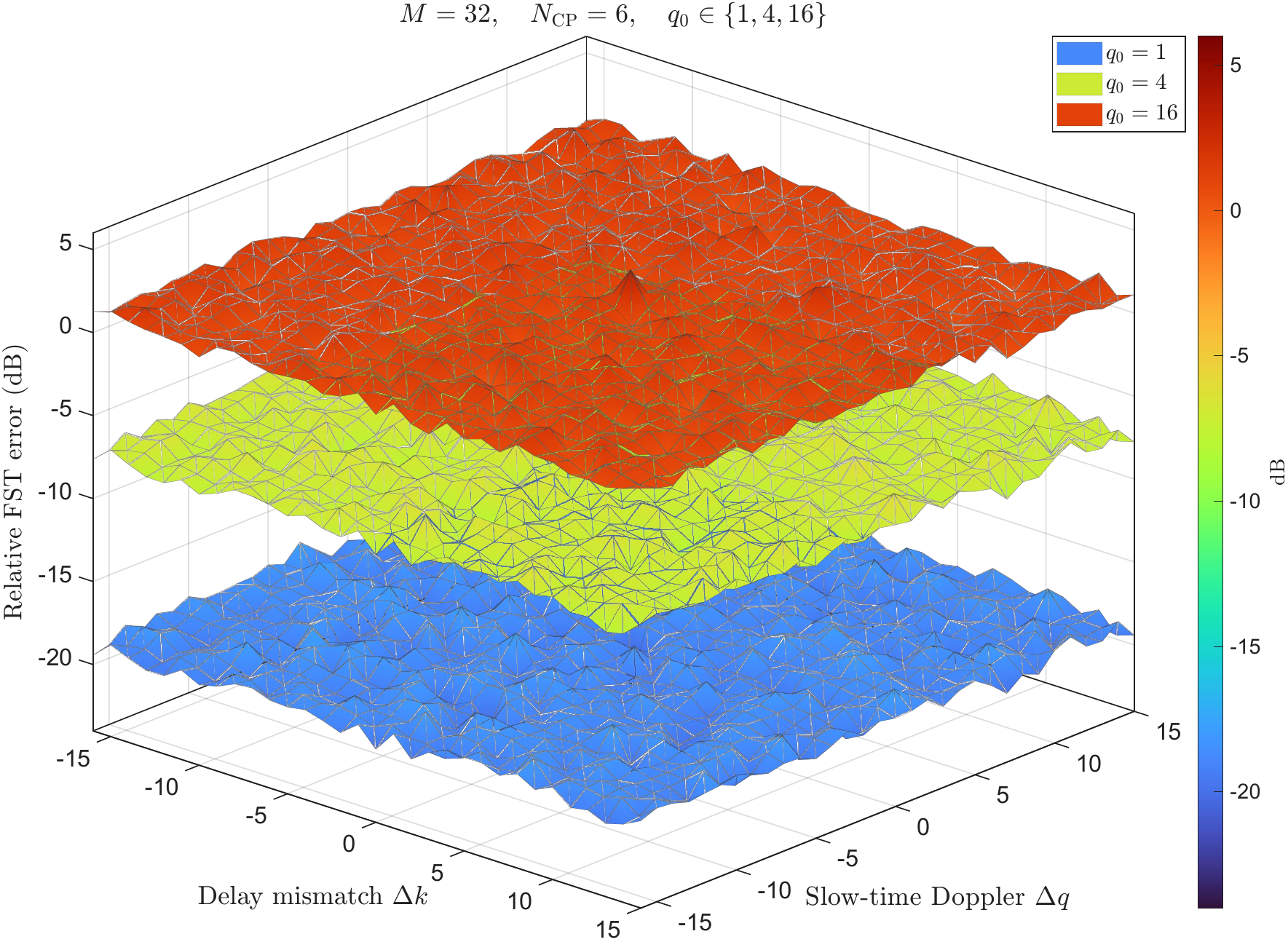}
  \caption{Pointwise relative FST approximation error
    $\varepsilon_{\rm AF}(\Delta k,\Delta q;q_0)$ in \eqref{eq:fst_af_approx_error} over the
    centered integer delay and slow-time Doppler mismatches.}
  \label{fig:fst_approximation_error}
\end{figure}

\subsection{Design Implications and Modeling Scope}

\subsubsection{Interpretation of EISL and Sidelobe Statistics}

Remark~\ref{rem:statistical_decomposition} separates the expected sidelobe power into a
coherent term, total kernel energy, and a fourth-order correction. EISL integrates these terms over
the prescribed 3D sidelobe region, while ESL divides by its measure. Both are ensemble
metrics: neither bounds the largest sidelobe of an individual payload realization. Such a
guarantee would require a separate analysis of peak statistics or exceedance probabilities.

Because the angular response is not translation invariant
(Remark~\ref{rem:angle_nontranslation}), EISL evaluated for a boresight target need not
describe targets elsewhere in the field of view. A sector-wide design can instead aggregate
the objective over the desired target angles, for example by an average or worst-case value.

\subsubsection{Role of Precoding and Constellation Shaping}

The subcarrier precoders enter the deterministic kernel
$\mathbf{M}_{k,q}(\theta,\vartheta)$, so EISL provides an objective for precoder design.
For comparisons at a fixed total transmit power, an illustrative formulation is
\begin{equation}
  \begin{aligned}
    \min_{\{\mathbf{W}_n\}}\quad
     & \int_{\Omega}
    \chi_{\rm sl}(k,q,\theta)
    \mathbb{E}
    \left[
      |\mathcal{A}_{3D}(k,q;\theta,\theta_0)|^2
      \right]\,\mathrm{d}\zeta
    \\
    \operatorname{s.t.}\quad
      & \sum_{n=0}^{N-1}
     \operatorname{Tr}\left(\mathbf{W}_n\mathbf{W}_n^H\right)
     = P.
  \end{aligned}
  \label{eq:eisl_precoder_problem}
\end{equation}
The equality fixes the total transmit power, so EISL comparisons among feasible
precoders are made at the same power level. In a practical ISAC design, this equality
can be combined with communication and hardware constraints such as rate or SINR,
per-antenna power, and spectral masks.
For fixed precoders and i.i.d.\ proper unit-power symbols, the constellation affects this
objective through the $(\kappa-2)$ correction to the diagonal kernel energy
(Remark~\ref{rem:statistical_decomposition}). This relation can guide constellation shaping
when the symbol assumptions remain valid. The precoders, by contrast, alter the full kernel
and couple the angular search cells. Since that kernel contains precoders on both sides of
$\boldsymbol{\Psi}_{k,q}$, the resulting objective is generally nonconvex. A deployment would
also impose its communication and hardware constraints, such as rate or SINR, per-antenna
power, and spectral masks. Solving the resulting joint design problem is beyond the scope
of this statistical characterization.

\subsubsection{Frequency-Dependent Extended Targets}

Appendix~\ref{app:extended_target} extends the single-symbol grid model to a deterministic
frequency-dependent target response and an assumed receiver template. These responses weight
the two sides of $\boldsymbol{\Psi}_{k,q}$, so template mismatch changes the AF kernel.
Conditional on the responses, Lemma~\ref{lem:quadratic_covariance} and
Theorem~\ref{thm:expected_squared_af} apply with the extended-target kernel. A stochastic
target response requires a further expectation over its distribution.

\subsubsection{Scope and Limitations}

The array model assumes narrowband, far-field steering common to all subcarriers; wideband
beam squint would require frequency-dependent steering in $\mathbf{G}_T(\theta)$. The cyclic
delay operator assumes block alignment and a CP long enough for the modeled delay. Delays
beyond the CP would introduce inter-block terms absent from the present kernel.

The moment formula also requires i.i.d.\ proper symbols with the prescribed fourth moment.
Pilots, bit loading, coding, or correlated streams can change the joint symbol moments;
\eqref{eq:expected_squared_af} would then need to be reevaluated under their actual
statistics. Clutter, array imperfections, and angle-dependent scattering are likewise outside
the point-target model. The deterministic extended-target result covers one form of frequency
selectivity, but random scattering requires an additional statistical model. Thus, the
single-symbol results apply to the specified integer-grid cyclic model, while the multi-symbol
FST results additionally rely on the small within-block phase condition in
\eqref{eq:fst_small_phase}.

\section{Conclusions}\label{sec:conclusions}

This paper investigated the ambiguity properties of random MIMO-OFDM communication
waveforms through a unified treatment of the delay-Doppler-angle 3D-AF and its
coherent multi-symbol FST extension. By modeling the point-target matched-filter
output as a quadratic form in i.i.d.\ proper unit-power payload symbols, we derived
closed-form expressions for the expected squared 3D-AF, EISL, and ESL\@. The resulting
deterministic AF matrix incorporates the delay-Doppler shifts, transmit and receive
array responses, and subcarrier-dependent precoders, thereby providing a common
framework for studying waveform statistics and spatial processing.

For the single-symbol 3D-AF, the expected squared response decomposes into a coherent
term, total kernel energy, and a fourth-order correction governed by the constellation
kurtosis. This decomposition shows that precoding changes the full AF matrix, whereas
the constellation affects the response through the diagonal kernel energy. Under the
considered settings, the sampled EISL increases from QPSK to 16-QAM, 256-QAM, and
proper Gaussian symbols, in accordance with their increasing fourth moments. The
peak-normalized ESL decreases with $N$ within the tested models and preserves the same
constellation ordering, while its single-symbol and FST levels must be compared within
their respective processing grids. The AF cuts further show that constellation effects
are most visible along delay for the adopted precoder and normalization, whereas the
Doppler and angle cuts nearly overlap. The angular response nevertheless depends on
the target direction through the array manifold and precoders, and therefore cannot
generally be described only by parameter offsets.

Under the FST formulation, the slow-time Doppler rotation is retained while the
within-block phase variation is neglected under the small-phase condition
$\eta\ll1$. When identical precoders are used across OFDM blocks, the FST AF matrix
has a Kronecker structure that separates the slow-time Doppler operator from the
delay-angle kernel. The corresponding sidelobe statistics retain the fourth-moment
dependence, and the numerical FST EISL and normalized ESL results follow the predicted
constellation ordering. A direct comparison with the serialized AF further shows that
the approximation error increases overall with $\eta$, clarifying the validity range
of the FST model.

Overall, the proposed framework clarifies how payload statistics, precoding, and array
responses jointly determine the sensing ambiguity of random MIMO-OFDM signals under
the stated cyclic-grid assumptions. Its conditional extension also accommodates
deterministic frequency-dependent target responses. Future work may extend the analysis
to communication-constrained fixed-power EISL precoder design, non-i.i.d.\ payloads,
stochastic target responses, and off-grid experimental validation.

\appendices

\section{Quadratic-Form Moments and Sidelobe Metrics}
\label{app:covariance}

Let
\begin{equation}
  u=\mathbf{s}^H\mathbf{P}\mathbf{s},
  \qquad
  v=\mathbf{s}^H\mathbf{Q}\mathbf{s}.
\end{equation}
Since $\mathbb{E}[\mathbf{s}\mathbf{s}^H]=\mathbf{I}$, we have
\begin{equation}
  \mathbb{E}\left[u\right]=\operatorname{Tr}\left(\mathbf{P}\right),
  \qquad
  \mathbb{E}\left[v\right]=\operatorname{Tr}\left(\mathbf{Q}\right).
\end{equation}
Moreover,
\begin{align}
  \mathbb{E}\left[uv^*\right]
   & =
  \sum_{i,j,k,l}
  [\mathbf{P}]_{i,j}[\mathbf{Q}]_{k,l}^*
  \mathbb{E}\left[s_i^*s_j s_k s_l^*\right].
\end{align}
To evaluate the fourth-order term, independence and zero mean restrict nonzero contributions
to paired indices. Properness, $\mathbb{E}[s_i^2]=0$, eliminates pairs of two unconjugated
or two conjugated symbols. The remaining conjugate pairings are
\begin{equation}
  i=j,\ k=l,
  \qquad
  i=k,\ j=l.
\end{equation}
When all four indices coincide, the two pairings contribute $2$ instead of the actual
fourth moment $\kappa$. Adding $(\kappa-2)$ for that case gives
\begin{equation}
  \begin{aligned}
    \mathbb{E}\left[s_i^*s_j s_k s_l^*\right]
     & =
    \delta_{ij}\delta_{kl}+\delta_{ik}\delta_{jl}\\[-1pt]
     & \quad+(\kappa-2)\delta_{ij}\delta_{jk}\delta_{kl}.
  \end{aligned}
\end{equation}
Here $\delta_{ij}$ is the Kronecker delta. Substitution yields
\begin{align}
  \mathbb{E}\left[uv^*\right]
   & =
  \sum_{i,j,k,l}
  [\mathbf{P}]_{i,j}[\mathbf{Q}]_{k,l}^*
  \delta_{ij}\delta_{kl}
  \nonumber \\
   & \quad+
  \sum_{i,j,k,l}
  [\mathbf{P}]_{i,j}[\mathbf{Q}]_{k,l}^*
  \delta_{ik}\delta_{jl}
  \nonumber \\
   & \quad+
  (\kappa-2)
  \sum_{i,j,k,l}
  [\mathbf{P}]_{i,j}[\mathbf{Q}]_{k,l}^*
  \delta_{ij}\delta_{jk}\delta_{kl}
  \nonumber \\
   & =
  \sum_{i,k}[\mathbf{P}]_{i,i}[\mathbf{Q}]_{k,k}^*
  \nonumber \\
   & \quad+
  \sum_{i,j}[\mathbf{P}]_{i,j}[\mathbf{Q}]_{i,j}^*
  \nonumber \\
   & \quad+
  (\kappa-2)\sum_i [\mathbf{P}]_{i,i}[\mathbf{Q}]_{i,i}^*
  \nonumber \\
    & =
   \tr(\mathbf{P})\tr(\mathbf{Q})^*
   +\tr(\mathbf{P}\mathbf{Q}^H)\nonumber\\[-1pt]
    & \quad+(\kappa-2)
   \sum_i [\mathbf{P}]_{i,i}[\mathbf{Q}]_{i,i}^*.
\end{align}
Since
\begin{equation}
  \operatorname{Cov}(u,v)=\mathbb{E}[uv^*]-\mathbb{E}[u]\mathbb{E}[v]^*,
\end{equation}
and
$\mathbb{E}\left[u\right]\mathbb{E}\left[v\right]^*=\operatorname{Tr}\left(\mathbf{P}\right)
\operatorname{Tr}\left(\mathbf{Q}\right)^*$, the product-of-means term cancels, leaving
\begin{equation}
  \begin{aligned}
    \operatorname{Cov}\left(u,v\right)
     & =
    \operatorname{Tr}\left(\mathbf{P}\mathbf{Q}^H\right) \\[-1pt]
     & \quad+(\kappa-2)
    \sum_i [\mathbf{P}]_{i,i}[\mathbf{Q}]_{i,i}^*.
  \end{aligned}
  \label{eq:cov_formula}
\end{equation}
This proves Lemma~\ref{lem:quadratic_covariance}. For the second moment in
Theorem~\ref{thm:expected_squared_af}, set
$\mathbf{P}=\mathbf{Q}=\mathbf{M}$ and add
$|\mathbb{E}[\mathbf{s}^H\mathbf{M}\mathbf{s}]|^2
  =|\operatorname{Tr}(\mathbf{M})|^2$, which gives
\begin{equation}
  \begin{aligned}
    \mathbb{E}\!\left[|\mathbf{s}^H\mathbf{M}\mathbf{s}|^2\right]
     & =
    \left|\operatorname{Tr}(\mathbf{M})\right|^2
    +\operatorname{Tr}(\mathbf{M}\mathbf{M}^H) \\
     & \quad+(\kappa-2)\sum_i|[\mathbf{M}]_{i,i}|^2.
  \end{aligned}
\end{equation}
Substituting $\mathbf{M}=\mathbf{M}_{k,q}(\theta,\vartheta)$ gives
Theorem~\ref{thm:expected_squared_af}. The sidelobe integrand is nonnegative, so
Tonelli's theorem permits expectation and integration to be interchanged:
  \begin{equation}
  \begin{aligned}
    \mathbb{E}\left[\ISL_{3D}(\mathbf{s};\theta_0)\right]
     & =
    \int_{\Omega}\chi_{\rm sl}(k,q,\theta)\\[-1pt]
     & \quad\times
    \mathbb{E}\!\left[
      |\mathcal{A}_{3D}(k,q;\theta,\theta_0)|^2
    \right]\,\mathrm{d}\zeta.
  \end{aligned}
\end{equation}
Substituting the second-moment expression yields \eqref{eq:EISL_closed_form}; division
by $\zeta(\Omega_{\rm sl}(\theta_0))$ gives the ESL in
Corollary~\ref{cor:eisl_esl}.

\section{Multi-Symbol and FST Derivations}
\label{app:multi_fst}

Following the FST construction in~\cite{zhang2025discreteaf}, let $\mathbf{x}_m$ denote
the CP-free time-domain vector of block $m$. With the CP insertion and removal matrices
$\mathbf{C}_{\rm cp}$ and $\mathbf{R}_{\rm cp}$ from Section~\ref{sec:system_model},
the serialized CPI vector is
\begin{equation}
 \widetilde{\mathbf{x}}:=
 \left[(\mathbf{C}_{\rm cp}\mathbf{x}_0)^T,\ldots,
 (\mathbf{C}_{\rm cp}\mathbf{x}_{M-1})^T\right]^T
 \label{eq:app_fst_serial}
\end{equation}
For target $\ell$, the Doppler phase at the sample indexed by $(m,n)$ is
\begin{equation}
 \begin{gathered}
 \resizebox{0.90\columnwidth}{!}{$\displaystyle
 e^{\frac{j2\pi q_{0,\ell}[m(N+N_{\rm CP})+n]}{M(N+N_{\rm CP})}}
 =e^{\frac{j2\pi q_{0,\ell}m}{M}}
 e^{\frac{j2\pi q_{0,\ell}n}{M(N+N_{\rm CP})}}$}\\[-1pt]
 \approx e^{\frac{j2\pi q_{0,\ell}m}{M}}.
 \end{gathered}
 \label{eq:app_fst_phase_factorization}
\end{equation}
The within-block factor can be replaced by a common phase if the parameter
$\eta$ defined in \eqref{eq:fst_small_phase} satisfies
\begin{equation}
  \eta\ll1.
 \label{eq:app_fst_condition}
\end{equation}
For $N_{\rm CP}\le N$, this condition also bounds the phase deviation across the entire
CP-extended block: every sample is at most $N$ positions from the first retained sample
$n=N_{\rm CP}$. Absorbing the phase at that reference into $\beta_\ell$, we approximate
the serialized Doppler operator by
\begin{equation}
 \mathbf{D}_{M(N+N_{\rm CP}),q_{0,\ell}}
 \approx
 \mathbf{D}_{M,q_\ell}\otimes\mathbf{I}_{N+N_{\rm CP}},
 \qquad q_\ell\in\mathbb{Z}_M.
 \label{eq:app_fst_operator}
\end{equation}
Here $q_\ell=\operatorname{mod}(q_{0,\ell},M)$. The approximation retains the slow-time
phase $e^{j2\pi m q_\ell/M}$ while discarding the within-block phase variation.

For an aligned OFDM block and an integer target delay $0\leq k_\ell\leq N_{\rm CP}$,
CP removal gives the cyclic-delay identity
\begin{equation}
 \mathbf{R}_{\rm cp}\,
 \mathbf{J}_{N+N_{\rm CP},k_\ell}\,
 \mathbf{C}_{\rm cp}
 =
 \mathbf{J}_{N,k_\ell}.
 \label{eq:app_cp_identity}
\end{equation}
Combining this identity with the Doppler approximation gives the CP-free received CPI model
\begin{equation}
 \mathbf{y}
 \approx
 \sum_{\ell=1}^{L}\beta_\ell
 \left(\mathbf{D}_{M,q_\ell}\otimes\mathbf{J}_{N,k_\ell}\right)
 \mathbf{x}
 +\mathbf{z},
 \label{eq:app_fst_vector_echo}
\end{equation}
where $\mathbf{x}$ and $\mathbf{y}$ stack the $M$ CP-free transmit and receive blocks.
Equivalently, the $m$th MIMO received block is
\begin{equation}
 \mathbf{Y}_m
 \approx
 \sum_{\ell=1}^{L}\beta_\ell e^{j2\pi m q_\ell/M}
 \mathbf{H}_s(\theta_\ell)\mathbf{X}_m
 \mathbf{J}_{N,k_\ell}^{T}
 +\mathbf{Z}_m.
 \label{eq:app_fst_mimo_echo}
\end{equation}

For a test cell $(k,q,\theta)$, the coherent matched-filter output is
\begin{equation}
 z_{\rm MF}^{\rm FST}(k,q,\theta)
 =
 \sum_{m=0}^{M-1}e^{-j2\pi mq/M}
 \operatorname{Tr}\!\left[
 \mathbf{X}_m^H\mathbf{H}_s^H(\theta)\mathbf{Y}_m\mathbf{J}_{N,k}
 \right].
 \label{eq:app_fst_mf}
\end{equation}
Substituting \eqref{eq:app_fst_mimo_echo} and retaining one target contribution yields
\begin{equation}
 \begin{aligned}
 \mathcal{A}_{3D}^{\rm FST}(k,q;\theta,\vartheta)
 &=\rho_r(\theta,\vartheta)\sum_{m=0}^{M-1}e^{-j2\pi mq/M}\\[-1pt]
 &\quad\times\mathbf{a}_t^H(\theta)\mathbf{X}_m^*\mathbf{J}_{N,k}^{T}
 \mathbf{X}_m^T\mathbf{a}_t(\vartheta),
 \end{aligned}
 \label{eq:app_fst_af}
\end{equation}
This is the FST cross-AF used in Section~\ref{sec:system_model}.

Using $\mathbf{X}_m^T\mathbf{a}_t(\theta)
=\mathbf{F}_N^H\mathbf{G}_{T,m}(\theta)\mathbf{s}_m$, define the delay-only kernel
\begin{equation}
  \boldsymbol{\Psi}_{k,0}:=\mathbf{F}_N\mathbf{J}_{N,k}^{T}\mathbf{F}_N^H.
  \label{eq:app_fst_psi}
\end{equation}
The cross-AF then takes the quadratic form
\begin{equation}
 \begin{aligned}
 \mathcal{A}_{3D}^{\rm FST}(k,q;\theta,\vartheta)
 &=\rho_r(\theta,\vartheta)\bar{\mathbf{s}}^H\bar{\mathbf{G}}_T^H(\theta)\\[-1pt]
 &\quad\times\left(\mathbf{D}_{M,q}^{*}\otimes\boldsymbol{\Psi}_{k,0}\right)
 \bar{\mathbf{G}}_T(\vartheta)\bar{\mathbf{s}}.
 \end{aligned}
 \label{eq:app_fst_quadratic}
\end{equation}
The quadratic form proves Proposition~\ref{prop:fst_matrix}. If the precoders are identical
across CPI blocks, the single-block kernel is independent of $m$, giving
\begin{equation}
 \bar{\mathbf{M}}_{k,q}^{\rm FST}(\theta,\vartheta)
 =
 \mathbf{D}_{M,q}^{*}\otimes\mathbf{M}_{k,0}(\theta,\vartheta),
 \label{eq:app_fst_kron}
\end{equation}
This proves Corollary~\ref{cor:fst_kron}. The fourth-order identity of
Appendix~\ref{app:covariance}, applied to \eqref{eq:app_fst_quadratic}, gives
Theorem~\ref{thm:fst_expected}. Integrating the resulting nonnegative second moment over
$\Omega_{{\rm FST},\rm sl}(\theta_0)$ gives Corollary~\ref{cor:fst_eisl}.

\section{Frequency-Dependent Extended-Target Formulation}
\label{app:extended_target}

Let an extended target have the deterministic subcarrier response
\begin{equation}
  \begin{aligned}
    \mathbf{b}_t
     & :=
    [b_{t,0},b_{t,1},\ldots,b_{t,N-1}]^T
    \in\CN^N, \\
    \mathbf{B}_t
     & :=
    \operatorname{diag}(\mathbf{b}_t),
  \end{aligned}
  \label{eq:extended_target_response}
\end{equation}
with circular time-domain filtering operator
\begin{equation}
  \boldsymbol{\Phi}_t
  :=
  \mathbf{F}_N^H\mathbf{B}_t\mathbf{F}_N.
  \label{eq:extended_target_filter}
\end{equation}
The receiver may use a different deterministic frequency response
$\mathbf{b}_{\rm ref}\in\CN^N$. Define
$\mathbf{B}_{\rm ref}=\operatorname{diag}(\mathbf{b}_{\rm ref})$ and
$\boldsymbol{\Phi}_{\rm ref}
  =\mathbf{F}_N^H\mathbf{B}_{\rm ref}\mathbf{F}_N$. For a target response
at the delay-Doppler origin and a test cell $(k,q,\theta)$, the
single-symbol grid-based conditional extended-target AF is
\begin{equation}
  \begin{aligned}
    \mathcal{A}_{3D}^{\rm e}
    (k,q;\theta,\vartheta\mid\mathbf{b}_{\rm ref},\mathbf{b}_t)
     & :=
    \left\langle
    \mathbf{H}_s(\theta)\mathbf{X}\mathbf{J}_{N,k}^T
    \boldsymbol{\Phi}_{\rm ref}^T\mathbf{D}_{N,q}^T,\right.\\[-1pt]
     & \qquad\left.
    \mathbf{H}_s(\vartheta)\mathbf{X}\boldsymbol{\Phi}_t^T
    \right\rangle_F.
  \end{aligned}
  \label{eq:extended_target_af}
\end{equation}
The cyclic-delay and target-filtering operators are circulant and commute; the Doppler
operator retains its position in the matched-filter product. Diagonalizing the circulant
operators and using \eqref{eq:x_theta} gives
\begin{equation}
  \begin{aligned}
     & \mathcal{A}_{3D}^{\rm e}
    (k,q;\theta,\vartheta
    \mid\mathbf{b}_{\rm ref},\mathbf{b}_t) \\[-1pt]
     & \quad=
    \mathbf{s}^H
    \mathbf{M}_{k,q}^{\rm e}
    (\theta,\vartheta
    \mid\mathbf{b}_{\rm ref},\mathbf{b}_t)
    \mathbf{s},
  \end{aligned}
  \label{eq:extended_target_quadratic}
\end{equation}
where
\begin{equation}
  \begin{aligned}
    \mathbf{M}_{k,q}^{\rm e}
    (\theta,\vartheta\mid\mathbf{b}_{\rm ref},\mathbf{b}_t)
     & :=
    \rho_r(\theta,\vartheta)\mathbf{G}_T^H(\theta)
    \mathbf{B}_{\rm ref}^H \\[-1pt]
     & \qquad\times\boldsymbol{\Psi}_{k,q}\mathbf{B}_t
    \mathbf{G}_T(\vartheta).
  \end{aligned}
  \label{eq:extended_target_kernel}
\end{equation}
The factors $\mathbf{B}_{\rm ref}^H$ and $\mathbf{B}_t$ explicitly represent receiver
template mismatch and target frequency selectivity. Conditional on these deterministic
responses, Theorem~\ref{thm:expected_squared_af} and
Corollary~\ref{cor:eisl_esl} apply after replacing $\mathbf{M}_{k,q}$ with
$\mathbf{M}_{k,q}^{\rm e}$. For a random target response, the resulting conditional
metrics require an outer expectation over $\mathbf{b}_t$.

\end{document}